\documentclass{statsoc}
\usepackage[a4paper]{geometry}
\usepackage{graphicx}
\usepackage{amsmath}
\usepackage{natbib}
\usepackage{pdfpages}

\usepackage{stmaryrd}
\usepackage{graphicx}
\usepackage{amssymb}
\usepackage{bbm}
\usepackage{verbatim}
\usepackage{enumitem}
\usepackage[hyperindex]{hyperref}
\usepackage{fancybox}
\usepackage{psfrag}

 \usepackage{array,multirow,graphicx}
 \usepackage{float}

\newcommand\rev[1]{#1}

\usepackage{array}
\usepackage{booktabs}
\newlength{\querylen}
\usepackage{tikz}
\usepackage{tikz-qtree}
\usetikzlibrary{bayesnet}

\newcommand{\R}{\mathbb{R}}

\newcommand{\E}{\mathbb{E}}
\newcommand{\I}{\mathbb{I}}

\newcommand{\var}{\textup{var}}

\newcommand{\X}{\textbf{X}}
\newcommand{\y}{\boldsymbol{y}}
\newcommand{\z}{\boldsymbol{z}}

\newtheorem{remark}{Remark}
\newtheorem{theorem}{Theorem}
\newtheorem{proposition}{Proposition}

\newtheorem{ass}{Assumption}

\newtheorem{lemma}{Lemma}
\newtheorem{coroll}{Corollary}

\newtheorem{?}{Question}

\usepackage{subfigure}

\newtheorem{taggedalgorithmx}{Algorithm}
\newenvironment{taggedalgorithm}[1]
 {\renewcommand\thetaggedalgorithmx{#1}\taggedalgorithmx}
 {\endtaggedalgorithmx}

\newcommand{\x}{\boldsymbol{x}}
\newcommand{\sX}{\mathcal{X}}
\newcommand{\tf}{\tilde{f}}

\usepackage{xr-hyper}
\title[Scalable Importance Tempering and Bayesian Variable Selection]{Scalable Importance Tempering and Bayesian Variable Selection}
\author[G.Zanella and G.Roberts]{Giacomo Zanella}
\address{Department of Decision Sciences, BIDSA and IGIER, Bocconi University,
Milan,
Italy.}
\author[G.Zanella and G.Roberts]{Gareth Roberts}
\address{Department of Statistics, University of Warwick,
Coventry,
United Kingdom.}

\begin{document}
\begin{abstract}
We propose a Monte Carlo algorithm to sample from high-dimensional probability distributions that combines Markov chain Monte Carlo (MCMC) and importance sampling.
We provide a careful theoretical analysis, including guarantees on robustness to high-dimensionality, explicit comparison with standard MCMC and illustrations of the potential improvements in efficiency. 
Simple and concrete intuition is provided for when the novel scheme is expected to outperform standard ones.
When applied to Bayesian Variable Selection problems, the novel algorithm is orders of magnitude more efficient than available alternative sampling schemes and allows to perform fast and reliable fully Bayesian inferences with tens of thousand regressors.
\end{abstract}

\section{Introduction}
Sampling from high-dimensional probability distributions is a common task arising in many scientific areas, such as Bayesian statistics, machine learning and statistical physics.
In this paper we propose and analyse a novel Monte Carlo scheme for generic, high-dimensional target distributions that combines importance sampling and Markov chain Monte Carlo (MCMC).


There have been many attempts to embed importance sampling within Monte Carlo schemes for Bayesian analysis, see for example \cite{smith1992bayesian,gramacy2010importance} and beyond. However, except where Sequential Monte Carlo approaches can be adopted, pure Markov chain based schemes (i.e. ones which simulate from precisely the right target distribution with no need for subsequent importance sampling correction) have been far more successful.
This is because MCMC methods are usually much more scalable to high-dimensional situations, see for example \citep{frieze1994sampling,belloni2009computational,Yang2016,
roberts2016complexity},
whereas importance sampling weight variances tend to grow (often exponentially) with dimension. 
In this paper we propose a natural way to combine the best of MCMC and importance sampling in a way that is robust in high-dimensional contexts and ameliorates the slow mixing which plagues many Markov chain based schemes.
\rev{The proposed scheme, which we call Tempered Gibbs Sampler (TGS), involves componentwise updating rather like Gibbs Sampling (GS), with improved mixing properties and associated importance
weights which remain stable as dimension increases.}
Through an appropriately designed tempering mechanism, TGS circumvents the main limitations of standard GS, such as the slow mixing induced by strong posterior correlations. 
It also avoids the requirement to visit all coordinates sequentially, instead iteratively making state-informed decisions as to which coordinate should be 
next updated.

Our scheme differentiates from classical simulated and parallel tempering \citep{marinari1992simulated,geyer1995annealing} in that it tempers only the coordinate that is currently being updated, and compensates for the overdispersion induced by the tempered update by choosing to update components which are in the tail of their conditional distributions more frequently.
The resulting dynamics can dramatically speed up convergence of the standard GS, both during the transient and the stationary phase of the algorithm.
Moreover, 
\rev{TGS does not require multiple temperature levels (as in simulated and parallel tempering) and thus avoids the tuning issues related to choosing the number of levels and collection of temperatures, as well as the heavy computational burden induced by introducing multiple copies of the original state space.} 
 
We apply the novel sampling scheme to Bayesian Variable selection problems, observing multiple orders of magnitude improvements compared to alternative Monte Carlo schemes. For example, TGS allows to perform reliable, fully Bayesian inference for spike and slab models with over ten thousand regressors in less than two minutes using a simple R implementation and a single desktop computer.

The paper structure is as follows.
The TGS scheme is introduced in Section \ref{sec:TGS}. There we provide basic validity results and intuition on the potential improvement given by the the novel scheme, together with an illustrative example.
In Section \ref{sec:theory} we develop a careful analysis of the proposed scheme. 
First we show that, unlike common tempering schemes, TGS is robust to high-dimensionality of the target as the coordinate-wise tempering mechanism employed is actually improved rather than damaged by high-dimensionality.
Secondly we show that TGS cannot perform worse than standard GS by more than a constant factor that can be chosen by the user (in our simulations we set it to 2), while being able to perform orders of magnitude better.
Finally we provide concrete insight regarding the type of correlation structures where TGS will perform much better than GS and the ones where GS and TGS will perform similarly. 
In Section \ref{sec:BVS} we provide a detailed application to Bayesian Variable selection problems\rev{, including computational complexity results.
Section \ref{sec:simulations} contains simulation studies.}
We review our findings in Section \ref{sec:discussion}.
\rev{Short proofs are directly reported in the paper, while longer ones can be found in the online supplementary material.}

\section{The Tempered Gibbs Sampling scheme}\label{sec:TGS}
Let $f(\x)$ be a probability distribution with $\x=(x_1,\dots,x_d)\in \sX_1\times\dots\times\sX_d=\sX$.
Each iteration of the classical random-scan Gibbs Sampler (GS) scheme proceeds by picking $i$ from $\{1,\dots,d\}$ uniformly at random and then sampling $x_i\sim f(x_i|\x_{-i})$.
We consider the following tempered version of the Gibbs Sampler, which depends on a collection of modified full conditionals denoted by $\{g(x_i|\x_{-i})\}_{i,\x_{-i}}$ with $i\in\{1,\dots,d\}$ and $\x_{-i}\in\sX_{-i}$.
The only requirement on $g(x_i|\x_{-i})$ is that, for all $\x_{-i}$, it is a probability density function on $\sX_i$ absolutely continuous with respect to $f(x_i|\x_{-i})$, with no need to be the actual full conditional of some global distribution $g(\x)$.
The following functions play a crucial role in the definition of the Tempered Gibbs Sampling (TGS) algorithm,
\begin{align}
p_i(\x)=
\frac{g(x_i|\x_{-i})}{f(x_i|\x_{-i})}
\quad\hbox{for }i=1,\dots,d\,;
\qquad
Z(\x)=\frac{1}{d}\sum_{i=1}^d p_i(\x)\,.
\end{align}
\begin{taggedalgorithm}{TGS}\label{alg:TGS}
At each iteration of the Markov chain do:
\begin{enumerate}[noitemsep,nolistsep]
\item\emph{(Coordinate selection)} Sample $i$ from $\{1,\dots,d\}$ proportionally to $p_i(\x)$.
\item\emph{(Tempered update)} Sample $x_i\sim g(x_i|\x_{-i})$.
\item\emph{(Importance weighting)} Assign to the new state $\x$ a weight 
$w(\x)=Z(\x)^{-1}$. 
\end{enumerate}
\end{taggedalgorithm}
The Markov chain $\x^{(1)},\x^{(2)},\dots$ induced by steps 1 and 2 of TGS is reversible with respect to $fZ$, which is a probability density function on $\sX$ defined as $(fZ)(\x)=f(\x)Z(\x)$.
We shall assume the following condition on $Z$ which is stronger than necessary, but which holds naturally for our purposes later on.
\begin{equation}\label{eq:Zbd}
\rev{
Z(\x)\hbox{ is bounded away from 0, and bounded above on compact sets.}
 } 
\end{equation}
\rev{Throughout the paper $Z$ and $w$ are the inverse of each other, i.e.\ $w(\x)=Z(\x)^{-1}$ for all $\x\in\sX$.
As usual, we denote the space of $f$-integrable functions from $\sX$ to $\R$ by $L^1(\sX,f)$ and we write $\E_f[h]=\int_{\sX}h(\x)f(\x)d\x$ for every $h\in L^1(\sX,f)$.
}
\begin{proposition}
$fZ$ is a probability density function on $\sX$ and the Markov chain $\x^{(1)},\x^{(2)},\dots$ induced by steps 1 and 2 of TGS is reversible with respect to $fZ$. Assuming that
(\ref{eq:Zbd}) holds and that TGS is $fZ$-irreducible, then
\begin{equation}\label{eq:ergodicity}
\hat{h}_{n}^{TGS}=\frac{\sum_{t=1}^n w(\x^{(t)}) h(\x^{(t)})}{\sum_{t=1}^n w(\x^{(t)})}
\to
\E_f[h]
\,,\qquad \hbox{as }n\to\infty\,,
\end{equation}
almost surely (a.s.) for every $h\in L^1(\sX,f)$.
\end{proposition}
\begin{proof}
Reversibility w.r.t.\ $f(\x)Z(\x)$ can be checked as in the proof of Proposition \ref{prop:wTGS} in Section \ref{appendix:wTGS} of the supplement.
Representing $f(\x)Z(\x)$ as a mixture of $d$ probability densities on $\sX$ we have
\begin{align*}
\int_{\sX}f(\x)Z(\x)d\x
=
\int_{\sX}\frac{1}{d}\sum_{i=1}^d f(\x)\frac{g(x_i|\x_{-i})}{f(x_i|\x_{-i})}d\x
=
\frac{1}{d}\sum_{i=1}^d\int_{\sX}f(\x_{-i})g(x_i|\x_{-i})d\x=1\,.
\end{align*}

The functions $h$ and $hw$ have identical support from (\ref{eq:Zbd}). Moreover it is clear that $h \in L^1 (\sX,f)$ if and only if $hw \in L^1 (\sX,fZ)$ and that in fact
$$
\E_f[h]=
\int h(\x ) f(\x ) d\x = \int h(\x) w(\x )  f(\x) Z(\x ) d\x=
\E_{fZ}[hw] \ .
$$
Therefore from Theorem 17.0.1 of \cite{.tweedie:1993:markov+stability+stochastic} applied to both numerator and denominator, (\ref{eq:ergodicity}) holds since by hypothesis TGS is $fZ$-irreducible so that $(\x^{(t)})_{t=1}^\infty$ is ergodic. 
\qed\end{proof}
We note that $fZ$-irreducibility of TGS can be established in specific examples using standard techniques, see for example 
\cite{.roberts:1994:conditions+algorithms+convergence}. Moreover under (\ref{eq:Zbd}) conditions from that paper which imply $f$-irreducibility of the standard Gibbs sampler readily extend to demonstrating that TGS is $fZ$-irreducible.

The implementation of TGS requires the user to specify a collection of densities $\{g(x_i|\x_{-i})\}_{i,\x_{-i}}$.
Possible choices of these include tempered conditionals of the form
\begin{equation}\label{eq:tempered_conditionals}
g(x_i|\x_{-i})=f^{(\beta)}(x_i|\x_{-i})=\frac{f(x_i|\x_{-i})^{\beta}}{\int_{\sX_i}f(y_i|\x_{-i})^{\beta}dy_i}\,,
\end{equation}
where $\beta$ is a fixed value in $(0,1)$, and mixed conditionals of the form
\begin{equation}\label{eq:mixed_conditionals}
g(x_i|\x_{-i})=\frac{1}{2}f(x_i|\x_{-i})+\frac{1}{2}f^{(\beta)}(x_i|\x_{-i})\,,
\end{equation}
with $\beta\in(0,1)$ and $f^{(\beta)}$ defined as in \eqref{eq:tempered_conditionals}.
Note that $g(x_i|\x_{-i})$ in \eqref{eq:mixed_conditionals} are not the full conditionals of $\frac{1}{2}f(\x)+\frac{1}{2}f^{(\beta)}(\x)$ as the latter would have mixing weights depending on $\x$. Indeed $g(x_i|\x_{-i})$ in \eqref{eq:mixed_conditionals} are unlikely to be the full conditionals of any distribution.

\rev{The theory developed in Section \ref{sec:theory} will provide insight into which choice for $g(x_i|\x_{-i})$  leads to effective Monte Carlo methods.
Moreover, we shall see that building $g(x_i|\x_{-i})$ as a mixture of $f(x_i|\x_{-i})$ and a flattened version of $f(x_i|\x_{-i})$, as in \eqref{eq:mixed_conditionals}, is typically a robust and efficient choice.}

\rev{The modified conditionals need to be tractable, as we need to sample from them and evaluate their density.
In many cases, 
if the original full conditionals $f(x_i|\x_{-i})$ are tractable (e.g.\ Bernoulli, Normal, Beta or Gamma distributions), then also the densities of the form $f^{(\beta)}(x_i|\x_{-i})$ are.
More generally, one can use any flattened version of $f(x_i|\x_{-i})$ instead of $f^{(\beta)}(x_i|\x_{-i})$.
For example in Section \ref{sec:gaussians} we provide an illustration using a $t$-distribution for $g(x_i|\x_{-i})$ when $f(x_i|\x_{-i})$ is normal.} 

TGS has various potential advantages over GS. First it makes an
``informed choice'' on which variable to update, choosing with higher probability coordinates whose value is currently in the tail of their conditional distribution. Secondly it induces potentially longer jumps by
sampling $x_i$ from a tempered distribution $g(x_i|\x_{-i})$. Finally, as we 
will see in the next sections, the invariant distribution $fZ$ has 
potentially much less correlation among variables compared to the original distribution $f$.

\subsection{Illustrative example.}\label{sec:illustrative_intro}
Consider the following illustrative example, where the target is a bivariate Gaussian with correlation $\rho=0.999$.
Posterior distributions with such strong correlations naturally arise in Bayesian modeling, e.g.\ in the context of hierarchical linear models with a large number of observations. 
The left of Figure \ref{fig:2dgaussian} displays the first 200 iterations of GS.
As expected, the strong correlation slows down the sampler dramatically and the chain hardly moves away from the starting point, in this case $(3,3)$.
The center and right of Figure \ref{fig:2dgaussian} display the first 200 iterations of TGS with modified conditionals given by \eqref{eq:tempered_conditionals} and \eqref{eq:mixed_conditionals}, respectively, and $\beta=1-\rho^2$.
See Section \ref{sec:theory} for some discussion on the choice fo $\beta$ in practice.
Now the tempered conditional distributions of TGS allow the chain to move freely around the state space despite correlation. 
However, the vanilla version of TGS, which uses tempered conditionals as in \eqref{eq:tempered_conditionals}, spends the majority of its time outside the region of high probability under the target.
\rev{This results in high variability of the importance weights $w(\x^{(t)})$ (represented by the size of the black dots in Figure \ref{fig:2dgaussian}), which deteriorates the efficiency of the estimators $\hat{h}^{TGS}_t$ defined in \eqref{eq:ergodicity}. 
On the other hand, the TGS scheme that uses tempered conditionals as in \eqref{eq:mixed_conditionals}, which we refer as TGS-mixed here, achieves both fast mixing of the Markov chain $\x^{(t)}$ and low variance of the importance weights $w(\x^{(t)})$.
For example, for the simulations of Figure \ref{fig:2dgaussian}, the estimated variances of the importance weights for TGS-vanilla and TGS-mixed are $16.2$ and $0.88$, respectively.}
In Section \ref{sec:theory} we provide theoretical analysis, as well as intuition, to explain the behaviour of TGS schemes.
\begin{figure}[h!]
\includegraphics[width=\linewidth]{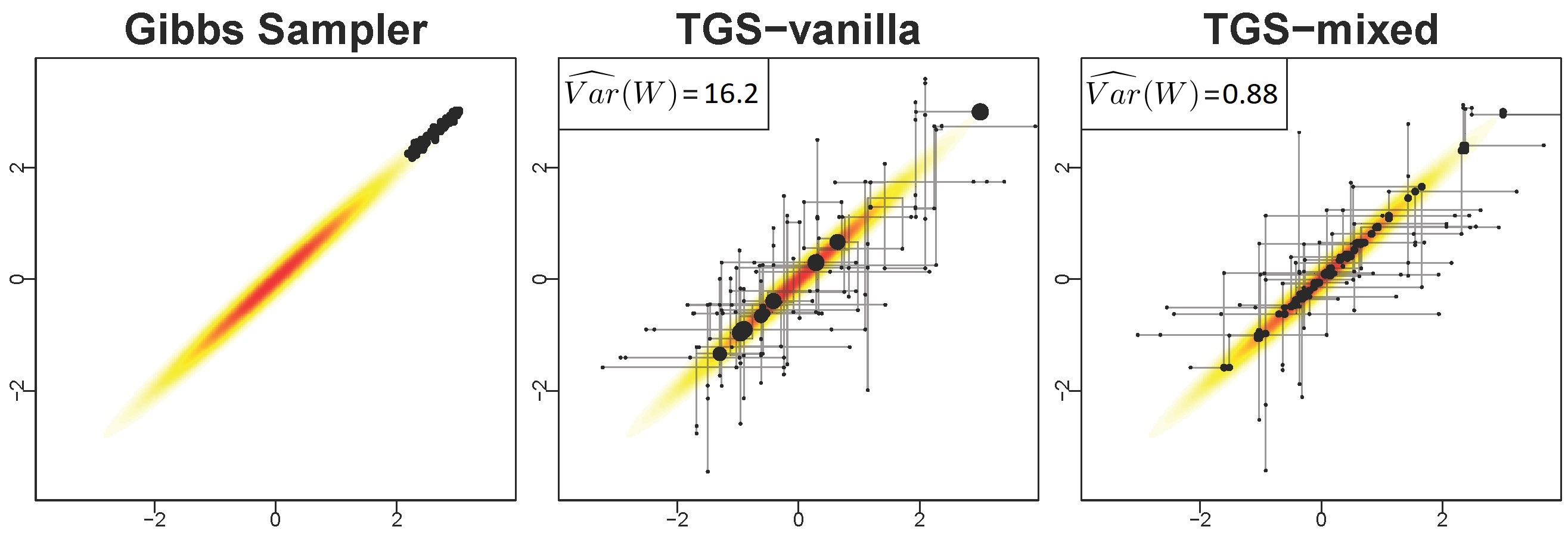}
\caption{\rev{Comparison of GS with two versions of TGS for $n=200$ iterations on a strongly correlated bivariate distribution. The sizes of the black dots are proportional to the importance weights $(w(\x^{(t)}))_{t=1}^n$.
$\widehat{Var}(W)$ refers to the estimated normalised variance of the importance weights, defined as $\widehat{Var}(W)=\frac{1}{n}\sum_{t=1}^n\bar{w}_t^2 -1$, where $\bar{w}_t=w(\x^{(t)})/(\frac{1}{n}\sum_{s=1}^n w(\x^{(s)}))$.
}}\label{fig:2dgaussian}
\end{figure}

\begin{remark}\label{rmk:invariance}
The TGS algorithm inherits the robustness and tuning-free properties of GS, such as invariance to coordinate rescalings or translations.
More precisely, the MCMC algorithms obtained by applying TGS to the original target $f(\x)$ or to the target obtained by applying any bijective transformation to a coordinate $x_i$ are equivalent, provided $g(x_i|\x_{-1})$ are also transformed accordingly.
A practical implication is that the TGS implementation does not require careful tuning of the scale of the proposal distribution such as typical Metropolis-Hasting algorithms do.
It is also trivial to see that TGS is invariant to permutations of the order of coordinates.
\end{remark}

\begin{remark}[Extended target interpretation]\label{rmk:extended_target}
The TGS scheme has a simple alternative construction that will be useful in the following.
Consider the extended state space $\sX \times \{1,\dots,d\}$ with augmented target
$$
\tf(\x,i)=\frac{1}{d}f(\x_{-i})g(x_i|\x_{-i})\qquad (\x,i)\in \sX \times \{1,\dots,d\}\,.
$$
The integer $i$ represents which coordinate of $\x$ is being tempered, and $g(x_i|\x_{-i})$ is the tempered version of $f(x_i|\x_{-i})$.
The extended target $\tf$ is a probability density function over $\sX \times \{1,\dots,d\}$ with marginals over $i$ and $\x$ given by 
\begin{align*}
\tf(i)
=&
\int\tf(\x,i)d\x
=
\frac{1}{d}
\\
\tf(\x)
=&
\sum_{i=1}^n \tf(\x,i)
=\frac{1}{d}\sum_{i=1}^df(\x_{-i})g(x_i|\x_{-i})
=f(\x)Z(\x)\,.
\end{align*}
TGS can be seen as a scheme that targets $\tf$ by alternating sampling from $\tf(i|\x)$ and $\tf(x_i|i,\x_{-i})$, and then corrects for the difference between $\tf$ and $f$ with $Z(\x)^{-1}$.
A direct consequence of this extended target interpretation is that the marginal distribution of $i$ is uniform, meaning that each coordinate gets updated every $1/d$ iterations on average.
\end{remark}

\section{Analysis of the algorithm}\label{sec:theory}
In this section we provide a careful theoretical and empirical analysis of the TGS algorithm.
The first aim is providing theoretical guarantees on the robustness of TGS, both in terms of variance of the importance sampling weights in high dimensions and mixing of the resulting Markov chain compared to the GS one.
The second aim is to provide understanding about which situations will be favourable to TGS and which one will not.
The main message is that the performances of TGS are never significantly worse than the GS ones while, depending on the situation, can be much better.

A key quantity in the discussion of TGS robustness is the following ratio between the original conditionals and the modified ones
\begin{equation}\label{eq:b}
b=\sup_{i,\x}
\frac{f(x_i|\x_{-i})}{g(x_i|\x_{-i})}\,.
\end{equation}
In order to ensure robustness of TGS, we want the constant $b$ to be finite and not too large.
This can be easily achieved in practice.
For example setting $g(x_i|\x_{-i})$ as in \eqref{eq:mixed_conditionals} we are guaranteed to have $b\leq 2$.
More generally, choosing $g(x_i|\x_{-i})=\frac{1}{1+\epsilon}f(x_i|\x_{-i})+\frac{\epsilon}{1+\epsilon}f^{(\beta)}(x_i|\x_{-i})$ we obtain $b\leq 1+\epsilon$.
The important aspect to note here is that \eqref{eq:b} involves only ratios of one-dimensional densities rather than $d$-dimensional ones (more precisely densities over $\sX_i$ rather than  over $\sX$).

\rev{
Throughout the paper, we measure the efficiency of Monte Carlo algorithms through their \emph{asymptotic variances}.
The smaller the asymptotic variance, the more efficient the algorithm.
For any $h\in L^2(\sX,f)$, the asymptotic variance associated to TGS is defined as 
$\var(h,TGS)=\lim_{n\to\infty}n\,\var(\hat{h}_{n}^{TGS})$, where $\hat{h}_{n}^{TGS}$ is the TGS estimator defined in \eqref{eq:ergodicity}.
The following lemma provides a useful representation of $\var(h,TGS)$.
}
\rev{\begin{lemma}\label{lemma:as_var_decomp}
Let $h\in L^1(\sX,f)$ and $\bar{h}(\x)=h(\x)-\E_f[h]$. 
If $\var(h,TGS)<\infty$ then
\begin{equation}\label{eq:as_var_decomp}
\var(h,TGS)
=
\E_f[\bar{h}^2w]
\left(1+2\sum_{t=1}^\infty\rho_t\right)\,,
\end{equation}
where $\rho_t$ is the lag-$t$ autocorrelation of $(w(\x^{(i)})\bar{h}(\x^{(i)}))_{i=1}^\infty$ and $(\x^{(i)})_{i=1}^\infty$ is the discrete-time chain induced by TGS started in stationarity.
\end{lemma}}
\rev{
The term $\E_f[\bar{h}^2w]$ in \eqref{eq:as_var_decomp} equals the asymptotic variance of the hypothetical importance sampler that uses $fZ$ as a proposal.
More formally, for any $h\in L^1(\sX,f)$ define the Self-normalised Importance Sampling (SIS) estimator as $\hat{h}_{n}^{SIS}=\frac{\sum_{i=1}^n w(\y^{(i)}) h(\y^{(i)})}{\sum_{i=1}^n w(\y^{(i)})}$, where $(\y^{(i)})_{i=1}^\infty$ is a sequence of independent and identically distributed (i.i.d.) random variables with distribution $fZ$.
Standard important sampling theory (see e.g. \citealp[Sec.3.2]{Deligiannidis2018}) tells us that $\E_f[\bar{h}^2w]=\var(h,SIS)$, where $\var(h,SIS)=\lim_{n\to\infty}n\,\var(\hat{h}_{n}^{SIS})$.
}
\rev{
Therefore the two terms in the right-hand side of \eqref{eq:as_var_decomp}, $\E_f[\bar{h}^2w]$ and $\left(1+2\sum_{t=1}^\infty\rho_t\right)$, can be interpreted as, respectively, the importance sampling and the MCMC contributions to $\var(h,TGS)$.
}

\subsection{Robustness to high-dimensionality.}
A major concern with classical importance tempering schemes is that they often collapse in high-dimensional scenarios (see e.g. \citealp[Sec.9.1]{Owen2013}).
The reason is that the ``overlap'' between the target distribution $f$ and a tempered version, such as $g=f^{(\beta)}$ with $\beta\in(0,1)$, can be extremely low if $f$ is a high-dimensional distribution.
On the contrary, the importance sampling procedure associated to TGS is robust to high-dimensional scenarios. 
\rev{This can be quantified by looking at the asymptotic variances $\var(h,SIS)=\E_f[\bar{h}^2w]$, or at the variance of the importance weights $W=w(\X)$ for $\X\sim fZ$. 
}
\rev{
\begin{proposition}\label{prop:as_var_IS_bounded}
Given $\X\sim fZ$ and $W=w(\X)$, we have
$$
Var\left(W\right)\leq b-1
\quad\hbox{and}\quad
\var(h,SIS)\leq b\,\var_f(h)\,,
$$
with $b$ defined in \eqref{eq:b} and $\var_f(h)=\E_f[h^2]-\E_f[h]^2$.
\end{proposition}
\begin{proof}
Equation \eqref{eq:b} implies $p_i(\x)\geq b^{-1}$ and thus $w(\x)=Z(\x)^{-1}\leq b$ for every $\x\in\sX$.
Combining the latter with $Var(W)=\E_{fZ}[w^2]-\E_{fZ}[w]^2=\E_{f}[w]-1$, we obtain $Var\left(W\right)=\E_{f}[w]-1\leq b-1$.
Again from $w(\x)\leq b$, we have
$\var(h,SIS)=\E_f[\bar{h}^2w]\leq b\E_f[\bar{h}^2]=b\,\var_f(h)$.
\qed\end{proof}
}
Proposition \ref{prop:as_var_IS_bounded} implies that, regardless of the dimensionality of the state space, \rev{the asymptotic variance $\var(h,SIS)$ is at most $b$ times $\var_f(h)$.}
\rev{Therefore, by \eqref{eq:as_var_decomp}, setting $b$ to a low value is sufficient to ensure that the importance sampling contribution to $\var(h,TGS)$ is well-behaved.
For example, if $g(x_i|\x_{-i})$ are chosen to be the mixed conditionals in \eqref{eq:mixed_conditionals} one is guaranteed to have $Var\left(W\right)\leq 1$ and $\var(h,SIS)\leq 2\var_f(h)$.
Note that the theoretical bound $Var\left(W\right)\leq 1$ is coherent with the estimated variance of the importance weights of TGS-mix in Figure \ref{fig:2dgaussian}.}

An even stronger property of TGS than the bounds in Proposition \ref{prop:as_var_IS_bounded} is that, under appropriate assumptions, $Var\left(W\right)$ converges to 0 as $d\to\infty$.
The underlying reason is that the weight function $w(\x)$ depends on an average of $d$ terms, namely $\frac{1}{d}\sum_{i=1}^d p_i(\x)$, and the increase of dimensionality has a stabilising effect on the latter. 
\rev{If, for example, the target has independent components with common distribution $f_0$, $f(\x)=\prod_{i=1}^d f_0(x_i)$, one can show that $Var\left(W\right)$ converges to 0 as $d\to\infty$.}
\rev{
\begin{proposition}\label{prop:ESS_limit}
Suppose $f(\x)=\prod_{i=1}^d f_0(x_i)$ and $g(x_i|\x_{-i})=g_0(x_i)$ where $f_0$ and $g_0$ are univariate probability density functions independent of $i$. If 
$\sup_{x_i} f_0(x_i)/g_0(x_i)\leq b<\infty$, then 
\begin{equation}\label{eq:var_to_0}
Var\left(W\right)\to 0 
\qquad\hbox{ as }d\to \infty\,.
\end{equation}
\end{proposition}
}
\rev{
\begin{proof}
By assumption we have $w(\x)^{-1}=\frac{1}{d}\sum_{i=1}^d\frac{g_0(x_i)}{f_0(x_i)}$.
Thus, given $\x\sim f$, $w(\x)^{-1}$ is the average of i.i.d.\ random variables with mean $1$ and converges a.s.\ to $1$ by the Strong Law of Large Numbers. It follows that $w(\x) \to 1$ a.s.\ as $d \to \infty$.
Also, $\sup_{x_i} f_0(x_i)/g_0(x_i)\leq b$ 
 implies $w(\x)=\left(\frac{1}{d}\sum_{i=1}^d\frac{g_0(x_i)}{f_0(x_i)}\right)^{-1}\leq b$. 
Thus by the Bounded Convergence Theorem $\E_f[w]\to 1$ as $d\to\infty$.
It follows $Var\left(W\right)=(\E_f[w]-1)\to 0$.
%
\qed\end{proof}
}
\rev{By contrast, recall that the importance weights associated to classical tempering (e.g.\ setting $g=f^{(\beta)}$ as importance distribution) in an i.i.d.\ context such as Proposition \ref{prop:ESS_limit} would have a variance growing exponentially with $d$ (see Examples 9.1-9.3 of \citet{Owen2013} for a more detailed discussion).}

\rev{
Proposition \ref{prop:ESS_limit} makes the assumption of independent and identically distributed components for simplicity and illustrative purposes. 
In fact, inspecting the proof of Proposition \ref{prop:ESS_limit}, one can see that \eqref{eq:var_to_0} holds whenever
$b<\infty$ and $\lim_{d\to\infty}\frac{1}{d}\sum_{i=1}^d p_i(\x)= 1$ in probability for $\x\sim f$.
Therefore, one could extend Proposition \ref{prop:ESS_limit} to any scenario where the law of large numbers for $\{p_i(\x)\}_{i}$ holds.
These include, for example, the case where $f$ has independent but non-identical components such that the variance of $p_i(\x)$ is bounded, i.e.\ $f(\x)=\prod_{i=1}^d f_i(x_i)$, $g(x_i|\x_{-i})=g_i(x_i)$ and $\int_{\sX_i}\frac{g_i(x_i)}{f_i(x_i)}g_i(x_i)dx_i$ bounded over $i$.
More generally, one could exploit laws of large numbers for dependent random variables in cases where the $d$ components of $\x\sim f$ enjoy some appropriate local dependence structure which is sufficient to have $\frac{1}{d}\sum_{i=1}^d p_i(\x)$ converging to a constant as $d\to\infty$.}

\subsection{Explicit comparison with standard Gibbs Sampling.}
We now compare the efficiency of the Monte Carlo estimators produced by TGS with the ones produced by classical GS.
For any function $h\in L^1(\sX,f)$ define the GS estimator of $\E_f[h]$ as $\hat{h}_{n}^{GS}=\frac{1}{n}\sum_{t=1}^n h(\y^{(t)})$, where $\y^{(1)},\y^{(2)},\dots$ is the $\sX$-valued Markov chain generated by GS, and denote the corresponding asymptotic variance by
$\var(h,GS)=\lim_{n\to\infty}n\,\var(\hat{h}_{n}^{GS})$.
The following theorem shows that the efficiency of TGS estimators can never be worse than the one of GS estimators by a factor larger than $b^2$.
\begin{theorem}\label{thm:TGS_GS_comparison}
For every $h\in L^2(\sX,f)$ we have
\begin{equation}\label{eq:peskun}
\var(h,TGS)
\leq 
b^2\var(h,GS)
+
b^2\var_f(h)\,.
\end{equation}
\end{theorem}
\begin{remark}\label{rmk:justalgsbd}
In most non-trivial scenarios, $\var_f(h)$ will be small in comparison to $\var(h,GS)$, because the the asymptotic variance obtained by GS is typically much larger than the one of an i.i.d.\ sampler.
In such cases we can interpret \eqref{eq:peskun} as saying that the asymptotic variance of TGS is at most $b^2$ times the ones of GS plus a smaller order term. 
\rev{More generally, since the Markov kernel associate to GS is a positive operator, we have $\var(h,GS)\geq \var_f(h)$ and thus, by \eqref{eq:peskun}, 
\begin{equation}
\label{eq:justalgsbd}
\var(h,TGS)\leq 2b^2\var(h,GS)\hbox{ for all }h\in L^2(\sX,f).
\end{equation}
}
\end{remark}
\rev{
\begin{remark}\label{rmk:finite_as_var}
Assuming $b<\infty$, Theorem \ref{thm:TGS_GS_comparison} implies that whenever $\var(h,GS)$ is finite then also $\var(h,TGS)$ is finite.
In general it is possible for $\var(h,TGS)$ to be finite when $\var(h,GS)$ is not.
The simplest example can be obtained setting $d=1$, in which case GS and TGS boil down to, respectively, i.i.d.\ sampling and importance sampling. In that case, any function $h$ such that $\int_\sX h(x)^2f(x)dx=\infty$ but $\int_\sX h(\x)^2w(\x)f(\x)d\x<\infty$ will satisfy $\var(h,GS)=\infty$ and $\var(h,TGS)<\infty$.
\end{remark}}

As discussed after equation \eqref{eq:b}, it is easy to set $b$ to a desired value in practice, for example using a mixture structure as in \eqref{eq:mixed_conditionals} which leads to the following corollary.
\begin{coroll}\label{coroll:peskun}
Let $\epsilon,\beta>0$.
If $g(x_i|\x_{-i})=\frac{1}{1+\epsilon}f(x_i|\x_{-i})+\frac{\epsilon}{1+\epsilon}f^{(\beta)}(x_i|\x_{-i})$ then
$$
\var(h,TGS)
\leq 
(1+\epsilon)^2\var(h,GS)
+
(1+\epsilon)^2\var_f(h)\,.
$$
\end{coroll}
By choosing $\epsilon$ to be sufficiently small, we have theoretical guarantees that TGS is not doing more than $(1+\epsilon)^2$ times worse than GS. 
Choosing $\epsilon$ too small, however, will reduce the potential benefit obtained with TGS, with the latter collapsing to GS for $\epsilon=0$, so that
optimising involves a compromise between these extremes.
The optimal choice involves a trade-off between small variance of the importance sampling weights and fast mixing of the resulting Markov chain.
In our examples we used $\epsilon=1$, leading to \eqref{eq:mixed_conditionals}, which is a safe and robust choice both in terms of importance sampling variance and of Markov chain mixing.

\subsection{TGS and correlation structure.}\label{sec:covariance structure}
Theorem \ref{thm:TGS_GS_comparison} implies that, under suitable choices of $g(x_i|\x_{-i})$, TGS never provides significantly worse (i.e.\ worse by more than a controllable constant factor) efficiency than GS.
On the other hand, TGS performances can be much better than standard GS.
The underlying reason is that the tempering mechanism can dramatically speed up the convergence of the TGS Markov chain $\x^{(t)}$ to its stationary distribution $fZ$ by reducing correlations in the target.
In fact, the covariance structure of $fZ$ is substantially different from the one of the original target $f$ and this can avoid the sampler from getting stuck in situations where GS would. 
\begin{figure}[h!]
\includegraphics[width=\linewidth]{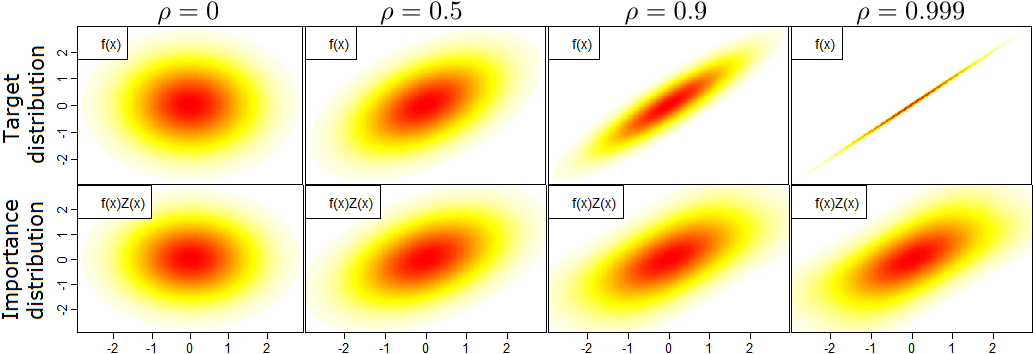}
\caption{Comparison between $f$ and $fZ$, first and second row respectively, for increasing correlation.
Here $f$ is a symmetric bivariate normal with correlation $\rho$ and $g=f^{(\beta)}$ with $\beta=1-\rho^2$.
}\label{fig:importance_distribution}
\end{figure}
Figure \ref{fig:importance_distribution} displays the original target $f$ and the modified one $fZ$ for a bivariate Gaussian with increasing correlation.
Here the modified conditionals are defined as in \eqref{eq:tempered_conditionals} with $\beta=1-\rho^2$.
It can be seen that, even if the correlation of $f$ goes to 1, the importance distribution $fZ$ does not collapse on the diagonal (note that $fZ$ is not Gaussian here).
As we show in the next section, this allows TGS to have a mixing time that is uniformly bounded over $\rho$.
Clearly, the same property does not hold for GS, whose mixing time deteriorates as $\rho\to 1$.

Note that a classical tempering approach would not help the Gibbs Sampler in this context. 
In fact, a Gibbs Sampler targeting $f^{(\beta)}$ with $\beta<1$ may be as slow to converge as one targeting $f$.
For example, in the Gaussian case the covariance matrix of $f^{(\beta)}$ is simply $\beta$ times the one of $f$ and thus, using the results of \citet{roberts1997updating}, a Gibbs Sampler targeting $f^{(\beta)}$ has exactly the same rate of convergence as one targeting $f$.
In the next section we provide some more rigorous understanding of the convergence behaviour of TGS to show the potential mixing improvements compared to GS.

\subsection{Convergence analysis in the bivariate case.}\label{sec:convergence}
In general, the TGS Markov chain $\x^{(t)}$ evolves according to highly complex dynamics and providing generic results on its rate of convergence of $fZ$ is extremely challenging.
Nonetheless, we now show that, using the notion of deinitialising chains from \citet{roberts2001markov} we can obtain rather explicit understanding of the convergence behaviour of $\x^{(t)}$ in the bivariate case.
The results suggest that, for appropriate choices of modified conditionals, the mixing time of $\x^{(t)}$ is uniformly bounded regardless of the correlation structure of the target.
This has to be contrasted with the chain induced by GS, 
whose mixing time diverges to infinity as the target's correlation goes to 1.

Our analysis proceeds as follows.
First we consider the augmented Markov chain $(\x^{(t)},i^{(t)})_{t=0}^\infty$ on $\sX\times \{1,\dots,d\}$ obtained by including the index $i$, as in Remark \ref{rmk:extended_target}. 
The transition from $(\x^{(t)},i^{(t)})$ to $(\x^{(t+1)},i^{(t+1)})$ is given by the following two steps:
\begin{enumerate}[noitemsep,nolistsep]
\item Sample $i^{(t+1)}$ from $\{1,\dots,d\}$ proportionally to $(p_1(\x^{(t)}),\dots,p_d(\x^{(t)}))$\,,
\item Sample $x_{i^{(t+1)}}^{(t+1)}\sim g(x_{i^{(t+1)}}|\x_{-i^{(t+1)}}=\x_{-i^{(t+1)}}^{(t)})$ and set $\x_{-i^{(t+1)}}^{(t+1)}=\x_{-i^{(t+1)}}^{(t)}$.
\end{enumerate}
Once we augment the space with $i^{(t)}$, we 
can ignore the component $x_{i^{(t)}}^{(t)}$, whose distribution is fully determined by $\x_{-i^{(t)}}^{(t+1)}$ and $i^{(t)}$.
More precisely, consider the stochastic process $(\z^{(t)},i^{(t)})_{t=0}^\infty$ obtained by taking
\begin{align*}
\z^{(t)}&=\x_{-i^{(t)}}^{(t)}\,,
&t\geq 0
\end{align*}
where $\x_{-i^{(t)}}^{(t)}$ denotes the vector $\x^{(t)}$ without the $i^{(t)}$-th component.
The following proposition shows that the process $(\z^{(t)},i^{(t)})_{t=0}^\infty$ is Markovian and contains all the information needed to characterise the convergence to stationarity of $\x^{(t)}$.
\begin{proposition}\label{prop:deinitialising}
The process $(\z^{(t)},i^{(t)})_{t=0}^\infty$ is a Markov chain and is deinitialising for $(\x^{(t)},i^{(t)})_{t=0}^\infty$, meaning that 
\begin{align}\label{eq:deinitialising}
\mathcal{L}(\x^{(t)},i^{(t)}|\x^{(0)},i^{(0)},\z^{(t)},i^{(t)})&=\mathcal{L}(\x^{(t)},i^{(t)}|\z^{(t)},i^{(t)})
&t\geq 1\,,
\end{align}
where $\mathcal{L}(\cdot|\cdot)$ denotes conditional distributions.
It follows that for any starting state $\x_*\in\sX$
\begin{align}\label{eq:deinitialising_TV}
\|
\mathcal{L}(\x^{(t)}|\x^{(0)}=\x_*)
-fZ\|_{TV}
=
\|
\mathcal{L}(\z^{(t)},i^{(t)}|\x^{(0)}=\x_*)
-\pi
\|_{TV}\,,
\end{align}
where $\|\cdot\|_{TV}$ denotes total variation distance and $\pi$ is the stationary distribution of $(\z^{(t)},i^{(t)})$.
\end{proposition}
Note that the conditioning on $\x^{(0)}$ in \eqref{eq:deinitialising_TV} is equivalent to conditioning on $(\x^{(0)},i^{(0)})$, because the distribution of $(\x^{(t)},i^{(t)})$ for $t>1$ is independent of $i^{(0)}$.

Proposition \ref{prop:deinitialising} implies that the convergence to stationarity of $\x^{(t)}$ is fully determined by that of $(\z^{(t)},i^{(t)})$.
In some situations, by looking at the chain $(\z^{(t)},i^{(t)})$ rather than 
$\x^{(t)}$, we can obtain a better understanding of the convergence properties 
of TGS. Consider for example the bivariate case, with $\sX=\R^2$ and target $f(x_1,x_2)$. 
In this context $(z^{(t)})_{t=0}^\infty$ is an $\R$-valued process, with stationary distribution $\frac{1}{2}f_1(z)+\frac{1}{2}f_2(z)$, where $f_1(z)=\int_{\R}f(z,x_2)dx_2$ and $f_2(z)=\int_{\R}f(x_2,z)dx_2$ are the target marginals.
In order to keep notation light and have results that are easier to interpret, here we further assume exchangeability, i.e.\ $f(x_1,x_2)=f(x_2,x_1)$, while Lemma \ref{lemma:bivariate} in the supplementary material online considers the generic case.
The simplification given by exchangeability is that it suffices to consider the Markov chain $(z^{(t)})_{t=0}^\infty$ rather than $(z^{(t)},i^{(t)})_{t=0}^\infty$.
\begin{proposition}\label{prop:bivariate_exchangeable}
Let $\sX=\R^2$ and $f$ be a target distribution with $f(x_1,x_2)=f(x_2,x_1)$, and marginal on $x_1$ denoted by $f_1$.
For any starting state $\x_*=(z_*,z_*)\in\R^2$ we have
$$\|
\mathcal{L}(\x^{(t)}|\x^{(0)}=\x_*)
-fZ\|_{TV}
=
\|
\mathcal{L}(z^{(t)}|z^{(0)}=z_*)
-f_1\|_{TV}\,,
$$
where $z^{(t)}$ is an $\R$-valued Markov chain with stationary distribution $f_1(z)$ and transition kernel
\begin{equation}\label{eq:trans_kernel_exch}
P(z'|z)=
r(z)\delta_{(z)}(z')+
q(z'|z)\alpha_b(z'|z)\,,
\end{equation}
where $r(z)=1-\int_{\R}\alpha_b(z'|z)q(z'|z)dz'$, $\alpha_b(z'|z)=\frac{f_1(z')q(z|z')}{f_1(z)q(z'|z)+f_1(z')q(z|z')}$ and $q(z'|z)=g(x_i=z'|x_{-i}=z)$.
\end{proposition}
The transition kernel in \eqref{eq:trans_kernel_exch} coincides with the one of an accept-reject algorithm with proposal distribution
$q(z'|z)=g(x_i=z'|x_{-i}=z)$ and acceptance given by the Barker rule, i.e.\ accept with probability $\alpha_b(z'|z)$.
The intuition behind the appearance of an accept-reject step is that updating the same coordinate $x_i$ in consequent iterations of TGS coincides with not moving the chain $(z^{(t)})$ and thus having a rejected transition.
Proposition \ref{prop:bivariate_exchangeable} implies that, given the modified conditionals $g(x_i|x_{-i})$, the evolution of $(z^{(t)})_{t=0}^\infty$ depends on $f$ only through the marginal distributions, $f_1$ or $f_2$, rather than on the joint distribution $f(x_1,x_2)$.

Proposition \ref{prop:bivariate_exchangeable} provides a rather complete understanding of TGS convergence behaviour for bivariate exchangeable distributions.
Consider for example a bivariate Gaussian target with correlation $\rho$, as in Section \ref{sec:illustrative_intro}. 
From Remark \ref{rmk:invariance}, we can assume without loss of generality $f$ to have standard normal marginals, and thus be exchangeable.
In this case $(z^{(t)})_{t=0}^\infty$ is a Markov chain with stationary distribution $f_1=N(0,1)$ and proposal $q(z'|z)=g(x_i=z'|x_{-i}=z)$.
For example, choosing modified conditionals as in \eqref{eq:tempered_conditionals} 
with $\beta=1-\rho^2$ we obtain $q(\cdot|z)=N(\rho z,1)$.
The worst case scenario for such a chain is $\rho=1$, where $q(\cdot|z)=N(z,1)$.
Nonetheless, even in this case the mixing of $(z^{(t)})_{t=0}^\infty$, and thus of $(\x^{(t)})_{t=0}^\infty$, does not collapse.
By contrast, the convergence of GS in this context deteriorates as $\rho\to 1$ as it is closely related to the convergence of the autoregressive process $z^{(t+1)}|z^{(t)}\sim N(\rho z, 1-\rho^2)$.
The latter discussion provides theoretical insight for the behaviour heuristically observed in Section \ref{sec:illustrative_intro}.
Proposition \ref{prop:bivariate_exchangeable} is not limited to the Gaussian context and thus we would expect that the qualitative behaviour just described holds much more generally.

\subsection{When does TGS work and when does it not?}\label{sec:gaussians}
The previous two sections showed that in the bivariate case TGS can induce much faster mixing compared to GS.
A natural question is how much this extends to the case $d>2$.
In this section we provide insight into when TGS substantially outperform GS and when instead they are comparable (we know by Theorem \ref{thm:TGS_GS_comparison} that TGS cannot converge substantially slower than GS).
The latter depends on the correlation structure of the target with intuition being as follows.
When sampling from a $d$-dimensional target $(x_1,\dots,x_d)$, the tempering mechanism of TGS allows to overcome strong pairwise correlations between any pair of variables $x_i$ and $x_j$ as well as strong $k$-wise \emph{negative} correlations, i.e.\ negative correlations between blocks of $k$ variables.
On the other hand, TGS does not help significantly in overcoming strong $k$-wise \emph{positive} correlations.
We illustrate this behaviour with a simulation study considering multivariate Gaussian targets with increasing degree of correlations (controlled by a parameter $\rho\in[0,1]$) under three scenarios.
Given the scale and translation invariance properties of the algorithms under consideration, we can assume w.l.o.g.\ the $d$-dimensional target to have zero mean and covariance matrix $\Sigma$ satisfying $\Sigma_{ii}=1$ for $i=1,\dots,n$ in all scenarios.
The first scenario considers pairwise correlation, with $d$ being a multiple of $2$ and $\Sigma_{2i-1,2i}=\rho$ for $i=1,\dots,\frac{d}{2}$ and $\Sigma_{ij}=0$ otherwise; the second exchangeable, positively-correlated distributions with $\Sigma_{ij}=\rho$ for all $i\neq j$;
the third exchangeable, negatively-correlated distributions with $\Sigma_{ij}=-\frac{\rho}{n-1}$ for all $i\neq j$.
In all scenarios, as $\rho\to 1$ the target distribution collapse to some singular distribution and the GS convergence properties deteriorate (see \citet{roberts1997updating} for related results). 

Figure \ref{fig:gaussian_simulations} reports the (estimated) asymptotic variance of the estimators of the coordinates mean (i.e.\ $h(\x)=x_i$, the value of $i$ is irrelevant) for $d=10$. 
We compare GS with two versions of TGS. 
The first has mixed conditionals as in \eqref{eq:mixed_conditionals},
with $\beta=1-\rho^2$.
Note that, by choosing a value of $\beta$ that depends on $\rho$ we are exploiting explicit global knowledge on $\Sigma$ in a potentially unrealistic way, matching the inflated conditional variance with the marginal variance.
Thus we also consider a more realistic situation where we ignore global knowledge on $\Sigma$ and set $g(x_i|\x_{-i})$ to be a t-distribution centred at $\E[x_i|\x_{-i}]$, with scale $\sqrt{\var(x_i|\x_{-i})}$ and shape $\nu=0.2$.
\begin{figure}[h!]
\includegraphics[width=\linewidth]{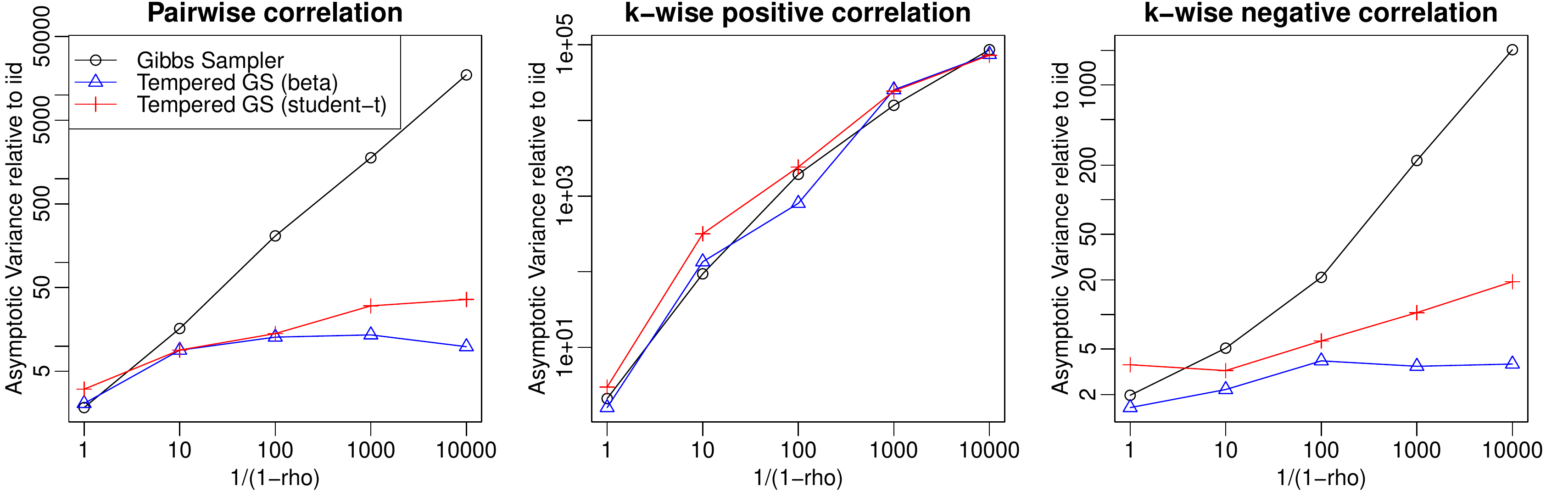}
\caption{Log-log plots of estimated asymptotic variances for GS compared to two versions of TGS on Gaussian targets with different covariance structures.}
\label{fig:gaussian_simulations}
\end{figure}
As expected, the asymptotic variance of the estimators obtained with GS deteriorate in all cases.
On the contrary, TGS performances do not deteriorate or deteriorate very mildly as $\rho\to 1$ for scenarios 1 and 3.
For scenario 2, TGS has very similar performances compared to GS.
In all cases, the two versions of TGS perform quite similarly, with the first of the two being slightly more efficient.
The qualitative conclusions of these simulations are not sensitive to various set-up details, such as: the value of $d$, the order of variables (especially in scenario 1) or the degree of symmetry.
Also, it is worth noting that TGS does not require prior knowledge of the global correlation structure or of which variable are strongly correlated to be implemented.

The reason for the presence or lack of improvements given by TGS lies in the different geometrical structure induced by positive and negative correlations.
Intuitively, we conjecture that if the limiting singular distribution for $\rho\to 1$ can be navigated with pairwise updates (i.e.\ moving on $(x_i,x_j)$ ``planes'' rather than $(x_i)$ ``lines'' as for GS), then TGS should perform well (i.e.\ uniformly good mixing over $\rho$ for good choice of $\beta$), otherwise it will not.

The intuition developed here will be useful in the Bayesian Variable Selection application of Section \ref{sec:BVS}, see for example the discussion in Section \ref{sec:complexity}.

\subsection{Controlling the frequency of coordinate updating.}\label{sec:wTGS_general}
In the extended target interpretation discussed in Remark \ref{rmk:extended_target} we have shown that the marginal distribution of $i$ under the extended target $\tilde{f}$ is uniform over $\{1,\dots,d\}$.
This implies that, for every $i,j\in\{1,\dots,d\}$, the TGS scheme will update $x_i$ and $x_j$ the same number of times on average.
In absence of prior information on the structure of the problem under consideration, the latter is a desirable robustness properties as it prevents the algorithm for updating some coordinates too often and ignoring others.
However, in some contexts, we may want to invest more computational effort in updating some coordinates rather than others (see for example the Bayesian Variable Selection problems discussed below).
This can be done by multiplying the selection probability $p_i(\x)$ for some weight function $\eta_i(\x_{-i})$, obtaining $p_i(\x)=\eta_i(\x_{-i})
\frac{g(x_i|\x_{-i})}{f(x_i|\x_{-i})}$ while leaving the rest of the algorithm unchanged.
We call the resulting algorithm weighted Tempered Gibbs Sampling (wTGS).
\begin{taggedalgorithm}{wTGS}
\label{alg:wTGS}
At each iteration of the Markov chain do:
\begin{enumerate}[noitemsep,nolistsep]
\item Sample $i$ from $\{1,\dots,d\}$ proportionally to 
$$
p_i(\x)=
\eta_i(\x_{-i})
\frac{g(x_i|\x_{-i})}{f(x_i|\x_{-i})}\,,
$$
\item Sample $x_i\sim g(x_i|\x_{-i})$,
\item Weight the new state $\x$ with a weight 
$Z(\x)^{-1}$ where \rev{$Z(\x)=\zeta^{-1}\sum_{i=1}^d p_i(\x)$ and $\zeta=\sum_{i=1}^d\E_{\x\sim f}[\eta_i(\x_{-i})]$.}
\end{enumerate}
\end{taggedalgorithm}
\rev{
The normalizing constant $\zeta$ in the latter definition of $Z(\x)$ is designed so that 
$\E_f[Z]=1$
 as for TGS.
When implementing wTGS, one needs to compute the weights $Z(\x)^{-1}$ only up to proportionality and thus $\zeta$ need not be computed explicitly.
TGS is a special case of wTGS obtained when $\eta_i(\x_{-i})= 1$, in which case $\zeta=d$.

As shown by the following proposition, the introduction of the weight functions $\eta_i(\x_{-i})$ does not impact the validity of the algorithm and it results in having a marginal distribution over the updated component $i$ proportional to $\E[\eta_i(\x_{-i})]$, where $\x\sim f$.
}
\begin{proposition}\label{prop:wTGS}
The Markov chain $\x^{(1)},\x^{(2)},\dots$ induced by steps 1 and 2 of wTGS is reversible with respect to $fZ$.
The frequency of updating of the $i$-th coordinate equals $\zeta^{-1}\mathbb{E}_{\x\sim f}[\eta_i(\x_{-i})]$.
\end{proposition}
By choosing $\eta_i(\x_{-i})$ appropriately, we can control the frequency with which we update each coordinate.
In Section \ref{sec:wTGS_BVS} we show an application of wTGS to Bayesian Variable Selection problems.

\section{Application to Bayesian Variable Selection}\label{sec:BVS}
We shall illustrate the theoretical and methodological conclusions of Section \ref{sec:theory} in an important class of statistical models where Bayesian computational issues are known to be particularly challenging.  
Binary inclusion variables in Bayesian Variable Selection models typically possess the kind of pairwise  and/or negative dependence structures conjectured to be conducive to successful application of TGS in Section \ref{sec:gaussians} (see Section \ref{sec:complexity} for a more detailed discussion).
Therefore, in this section we provide a detailed application of TGS to sampling from the posterior distribution of Gaussian Bayesian Variable Selection models.
This is a widely used class of models where posterior inferences are computationally challenging due to the presence of high-dimensional discrete parameters. 
In this context, the Gibbs Sampler is the standard choice of algorithm to draw samples from the posterior distribution (see Section \ref{appendix:BVS_lit_review} in the supplement for more details).

\subsection{Model specification.}\label{sec:BVS_model}
Bayesian Variable Selection (BVS) models provide a natural and coherent framework to select a subset of explanatory variables in linear regression contexts \citep{chipman2001practical}.
In standard linear regression, an $n\times 1$ response vector $Y$ is modeled as $Y|\beta,\sigma^2\sim N(X\beta,\sigma^2)$, where $X$ is an $n\times p$ design matrix and $\beta$ an $n\times 1$ vector of coefficients.
In BVS models a vector of binary variables $\gamma=(\gamma_1,\dots,\gamma_p)\in\{0,1\}^p$ is introduced to indicate which regressor is included in the model and which one is not ($\gamma_i=1$ indicates that the $i$-th regressor is included in the model and $\gamma_i=0$ that it is excluded).
The resulting model can be written as
\begin{align*}
Y|\beta_\gamma,\gamma,\sigma^2\sim& N(X_\gamma\beta_\gamma,\sigma^2\I_n)
\\
\beta_\gamma|\gamma,\sigma^2 \sim& N(0,\sigma^2 \Sigma_\gamma)\,
\\
p(\sigma^2)\propto& \frac{1}{\sigma^2}\,,
\end{align*}
where $X_\gamma$ is the $n\times |\gamma|$ matrix containing only the included columns of the $n\times p$ design matrix $X$, $\beta_\gamma$ is the $ |\gamma| \times 1$ vector containing only the coefficients corresponding the selected regressors and $\Sigma_\gamma$ is the $|\gamma|\times |\gamma|$ prior covariance matrix for the $|\gamma|$ selected regressors.
Here $|\gamma|=\sum_{i=1}^p\gamma_i$ denotes the number of ``active'' regressors.
The covariance $\Sigma_\gamma$ is typically chosen to be equal to a positive multiple of $(X_\gamma^TX_\gamma)^{-1}$ or the identity matrix, i.e.\ $\Sigma_\gamma=c(X_\gamma^TX_\gamma)^{-1}$ or $\Sigma_\gamma=c\I_{|\gamma|}$ for fixed $c>0$.
The binary vector $\gamma$ is given a prior distribution $p(\gamma)$ on $\{0,1\}^p$, for example assuming 
\begin{align*}
\gamma_i|h \stackrel{iid}\sim& \,\hbox{Bern}(h)
\qquad i=1,\dots,p\,,
\end{align*}
where $h$ is a prior inclusion probability, which can either be set to some fixed value in $(0,1)$ or be given a prior distribution (e.g.\ a distribution belonging to the Beta family).
\begin{remark}
One can also add an intercept to the linear model obtaining $Y|\beta_\gamma,\gamma,\sigma^2,\alpha\sim N(\alpha+X_\gamma\beta_\gamma,\sigma^2)$.
If such intercept is given a flat prior, $p(\alpha)\propto 1$, the latter is equivalent to centering $Y$, $X_1$, \dots, $X_p$ to have zero mean \citep[Sec.3]{chipman2001practical}.
\end{remark}
Under this model set-up, the continuous hyperparameters $\beta$ and $\sigma$ can be analytically integrated and one is left with an explicit expression for $p(\gamma|Y)$.
Sampling from such $\{0,1\}^p$-valued distribution allows to perform full posterior inferences for the BVS models specified above since $p(\beta_\gamma,\gamma,\sigma^2|Y)=p(\beta_\gamma,\sigma^2|\gamma,Y)p(\gamma|Y)$ and $p(\beta_\gamma,\sigma^2|\gamma,Y)$ is analytically tractable.
The standard way to draw samples from $p(\gamma|Y)$ is by performing Gibbs Sampling on the $p$ components $(\gamma_1,\dots,\gamma_p)$, repeatedly choosing $i\in\{1,\dots,p\}$ either in a random or deterministic scan fashion and then updating $\gamma_i\sim p(\gamma_i|Y,\gamma_{-i})$.

\subsection{TGS for Bayesian Variable Selection.}\label{sec:TGS_BVS}
We apply TGS to the problem of sampling from $\gamma\sim p(\gamma|Y)$.
\rev{Under the notation of Section \ref{sec:TGS}, this corresponds to $d=p$, $\sX=\{0,1\}^p$ and $f(\gamma)=p(\gamma|Y)$.}
For every value of $i$ and $\gamma_{-i}$, we set the tempered conditional distribution $g_i(\gamma_i|\gamma_{-i})$ to be the uniform distribution over $\{0,1\}$. 
It is easy to check that the supremum $b$ defined in \eqref{eq:b} is upper bounded by $2$ and thus we have theoretical guarantees on the robustness of TGS from Proposition \ref{prop:as_var_IS_bounded} and Theorem \ref{thm:TGS_GS_comparison}.

Since the target state space is discrete, it is more efficient to replace the Gibbs step of updating $\gamma_i$ conditional on $i$ and $\gamma_{-i}$, with its Metropolised version (see e.g. \citealp{liu1996peskun}).
The resulting specific instance of TGS is the following.
\begin{taggedalgorithm}{TGS for BVS}\label{alg:TGS_BVS}
At each iteration of the Markov chain do:
\begin{enumerate}[noitemsep,nolistsep]
\item Sample $i$ from $\{1,\dots,p\}$ proportionally to 
$
p_i(\gamma)=
\frac{1}{2p(\gamma_i|\gamma_{-i},Y)}
$.
\item Switch $\gamma_i$ to $1-\gamma_i$.
\item Weight the new state $\gamma$ with a weight 
$Z(\gamma)^{-1}$ where $Z(\gamma)=\frac{1}{p}\sum_{i=1}^p p_i(\gamma)$.
\end{enumerate}
\end{taggedalgorithm}
In step 1 above, $p(\gamma_i|\gamma_{-i},Y)$ denotes the probability that $\gamma_i$ takes its current value conditional on the current value of $\gamma_{-i}$ and on the observed data $Y$.
\rev{In the remainder of Section \ref{sec:BVS}, the expression TGS will refer to this specific implementation of the generic scheme described in Section \ref{sec:TGS}, and $P_{TGS}$ to the Markov transition kernel of the resulting discrete-time chain $(\gamma^{(t)})_{t=1}^\infty$.}

\subsection{wTGS for BVS.}\label{sec:wTGS_BVS}
As discussed in Section \ref{sec:wTGS_general}, TGS updates each coordinate with the same frequency. 
In a BVS context, however, this may be inefficient as the resulting sampler would spend most iterations updating variables that have low or negligible posterior inclusion probability, especially when $p$ gets large.
A better solution would be to update more often components with a larger inclusion probability, thus having a more focused computational effort.
In the wTGS framework of Section \ref{sec:wTGS_general}, this can be obtained using non-uniform weight functions $\eta_i(\gamma_{-i})$.
\rev{For example, Proposition \ref{prop:wTGS} implies that choosing $\eta_i(\gamma_{-i})=p(\gamma_i =1|\gamma_{-i},Y)$ leads to a frequency of updating of the $i$-th component 
equal to  $\zeta^{-1}\E[\eta_i(\gamma_{-i})]=s^{-1}p(\gamma_i =1|Y)$, where $s=\sum_{j=1}^p p(\gamma_j =1|Y)$ is the expected number of active variables a posteriori.}
Here $p(\gamma_i =1|Y)$ denotes the (marginal) posterior probability that $\gamma_i$ equals 1, while $p(\gamma_i=1|\gamma_{-i},Y)$ denotes the probability of the same event conditional on both the observed data $Y$ and the current value of $\gamma_{-i}$.
Note that with wTGS one can obtain a frequency of updating of the $i$-th component proportional to $p(\gamma_i =1|Y)$ without knowing the actual value of $p(\gamma_i =1|Y)$, but rather using only the conditional expressions $p(\gamma_i =1|\gamma_{-i},Y)$.

The optimal choice of frequency of updating is related to an exploration versus exploitation trade-off.
For example, choosing a uniform frequency of updating favours exploration, as it forces the sampler to explore new regions of the space by flipping variables with low conditional inclusion probability.
On the other hand, choosing a frequency of updating that focuses on variables with high conditional inclusion probability favours exploitation, as it allows the sampler to focus on the most important region of the state space. 
For this reason, we use a compromise between the choice of $\eta_i(\gamma_{-i})$ described above and the uniform TGS, obtained by setting $\eta_i(\gamma_{-i})=p(\gamma_i =1|\gamma_{-i},Y)+\frac{k}{p}$ with $k$ being a fixed parameter (in our simulations we used $k=5$).
Such choice leads to frequencies of updating given by a mixture of the uniform distribution over $\{1,\dots,p\}$ and the distribution proportional to $p(\gamma_i =1|Y)$.
More precisely we have $\zeta^{-1}\E[\eta_i(\gamma_{-i})]=\alpha\frac{p(\gamma_i =1|Y)}{s}+(1-\alpha)\frac{1}{p}$, where $\alpha=\frac{s}{k+s}$.
The resulting scheme is the following (see above for the definition of $p(\gamma_i =1|\gamma_{-i},Y)$).
\begin{taggedalgorithm}{wTGS for BVS}\label{alg:wTGS_BVS}
At each iteration of the Markov chain do:
\begin{enumerate}[noitemsep,nolistsep]
\item Sample $i$ from $\{1,\dots,p\}$ proportionally to 
$
p_i(\gamma)=
\frac{p(\gamma_i=1|\gamma_{-i},Y)+k/p}{2p(\gamma_i|\gamma_{-i},Y)}
$.
\item Switch $\gamma_i$ to $1-\gamma_i$.
\item Weight the new state $\gamma$ with a weight 
$Z(\gamma)^{-1}$ where $Z(\gamma)\propto\sum_{i=1}^p p_i(\gamma)$.
\end{enumerate}
\end{taggedalgorithm}
\rev{In the remainder of Section \ref{sec:BVS}, the expression wTGS will refer to this specific implementation of the generic scheme described in Section \ref{sec:wTGS_general}, and $P_{wTGS}$ to the Markov transition kernel of the resulting discrete-time Markov chain $(\gamma^{(t)})_{t=1}^\infty$.}

\subsection{Efficient implementation and Rao-Blackwellisation.}\label{sec:efficient_impl}
Compared to GS, TGS and wTGS provide substantially improved convergence properties at the price of an increased computational cost per iteration.
The additional cost is computing $\{p(\gamma_i|Y,\gamma_{-i})\}_{i=1}^p$ given $\gamma\in\{0,1\}^p$, which can be done efficiently through vectorised operations as described in Section \ref{appendix:implementation} of the supplement.
Such efficient implementation is crucial to the successful application of these TGS schemes.
\rev{The resulting cost per iteration of TGS and wTGS is of order $\mathcal{O}(np+|\gamma|p)$ .
For comparison, the cost per iteration of GS is $\mathcal{O}(n|\gamma|+|\gamma|^2)$.
If $X^TX$ has been precomputed before running the MCMC, then the costs per iteration become $\mathcal{O}(|\gamma|p)$ for TGS and $\mathcal{O}(|\gamma|^2)$ for GS.
In both cases, the relative additional cost of TGS over GS is $\mathcal{O}(p/|\gamma|)$.
See Section \ref{appendix:cost_iter} of the supplement for derivations of these expressions.
}

Interestingly, $\{p(\gamma_i|Y,\gamma_{-i})\}_{i=1}^p$ are the same quantity needed to compute Rao-Blackwellised estimators of the marginal Posterior Inclusion Probabilities (PIPs) $\{p(\gamma_i=1|Y)\}_{i=1}^p$.
Therefore, using TGS allows to implement Rao-Blackwellised estimators of PIPs (for all $i\in\{1,\dots,p\}$ at each flip) without extra cost. See Section \ref{appendix:RaoBlackwell} of the supplement for more details.

\rev{
\subsection{Computational complexity results for simple BVS scenarios}\label{sec:complexity}
In this section we provide quantitative results on the computational complexity of GS, TGS and wTGS in some simple BVS scenarios.
In particular, we consider two extreme cases, one where all regressors in the design matrix $X$ are orthogonal to each other (Section \ref{sec:diagonal}), and one where some of the regressors are perfectly collinear (Section \ref{sec:collinear}).
In the first case the posterior distribution $p(\gamma|Y)$ features  independent components and thus it is the ideal case for GS, while the second case it features some maximally correlated components and thus it is a worst-case scenario for GS.
Our results show that the computational complexity of TGS and wTGS is not impacted by the change in correlation structure between the two scenarios.
}\rev{
This is coherent with the conjecture of Section \ref{sec:gaussians} that the convergence of TGS and wTGS is not slowed down by pairwise and/or negative correlation.
In fact, a block of collinear regressors in the design matrix $X$ induces a corresponding block of \emph{negatively} correlated inclusion variables in $p(\gamma|Y)$.
See Section \ref{app:target_derivation} of the supplement for a quantitative example.
More generally, strong correlation among regressors induces strong \emph{negative} correlation among the corresponding inclusion variables in $p(\gamma|Y)$.
Intuitively, strongly correlated regressors provide the same type of information regarding $Y$.
Thus, conditional on the $i$-th regressor being included in the model, the regressors strongly correlated with the $i$-th one are not required to further explain the data and thus have a low probability of being included.
The latter holds regardless of whether the original correlation among regressors is positive or negative.

As a preliminary step for the results in Sections \ref{sec:diagonal} and \ref{sec:collinear}, we now discuss the definition of computational complexity that we will use.
}

\rev{
\subsubsection{Computational complexity for MCMC and importance tempering}\label{sec:complexity_definition}
In classical contexts, one can define the computational complexity of an MCMC algorithm as the product between the cost per iteration and the number of iterations required to obtain Monte Carlo estimators with effective sample size of order 1. 
One way to define such number of iterations is the so-called relaxation time, which is defined as the inverse of the spectral gap associated to the Markov kernel under consideration 
(for instance the second largest eigenvalue in the case where the Markov kernel has a purely discrete spectrum).
Such definition is motivated by the fact that the asymptotic variances associated to an $f$-reversible Markov kernel $P$ satisfy
\begin{align}\label{eq:discr_Gap}
\var(h,P)&\leq \frac{2\var_f(h)}{Gap(P)}
&h\in L^2(\sX,f)
\,,
\end{align}
where $Gap(P)$ is the spectral gap of $P$ \citep[Prop.1]{rosenthal2003asymptotic}.
Note that here $Gap(P)$ refers to the spectral gap of $P$ and not the \emph{absolute} spectral gap, see \citet{rosenthal2003asymptotic} for more discussion.
In the following we denote the relaxation time of GS as $t_{GS}=Gap(P_{GS})^{-1}$.
By \eqref{eq:discr_Gap}, we can interpret $2t_{GS}$ as the number of GS iterations required to have effective sample size equal to 1.
}
%

\rev{
On the other hand, TGS asymptotic variances include also an importance sampling contribution, see \eqref{eq:as_var_decomp}.
Thus the direct analogous of \eqref{eq:discr_Gap}, i.e.\  $\var(h,TGS)\leq 2Gap(P_{TGS})^{-1}\var_f(h)$, does not hold anymore and defining the TGS relaxation time as $Gap(P_{TGS})^{-1}$ would be inappropriate.
As shown by the following lemma, the problem can be circumvented using the spectral gap of a continuous-time version of TGS. 
In order to simplify the lemma's proof and notation, we assume $|\sX|<\infty$, which always holds in the BVS context.
We expect an analogous result to hold in the context of general state spaces $\sX$.
}\rev{
\begin{lemma}\label{lemma:cont_Gap}
Let $|\sX|<\infty$.
Define the jump matrix $Q_{TGS}$ on $\sX$ as $Q_{TGS}(\gamma,\gamma')=Z(\gamma)P_{TGS}(\gamma,\gamma')$ for all $\gamma'\neq\gamma$ and $Q_{TGS}(\gamma,\gamma)=-\sum_{\gamma'\neq\gamma} Q_{TGS}(\gamma,\gamma')$.
Then 
\begin{align}\label{eq:cont_Gap}
\var(h,TGS)\leq& \frac{2\,\var_f(h)}{Gap(Q_{TGS})}
&h:\sX\to\R
\,,
\end{align}
where $Gap(Q_{TGS})$ is the smallest non-zero eigenvalue of $-Q_{TGS}$.
\end{lemma}}\rev{
Lemma \ref{lemma:cont_Gap} implies that $Gap(Q_{TGS})$ implicitly incorporates both the importance sampling and the autocorrelation terms in $\var(h,TGS)$. 
Motivated by \eqref{eq:cont_Gap}, we define the relaxation time of TGS as $t_{TGS}=Gap(Q_{TGS})^{-1}$.
By Lemma \ref{lemma:cont_Gap}, one can still interpret $2t_{TGS}$ as the number of TGS iterations required to have effective sample size equal to 1.
Similarly, we define the relaxation time of wTGS as the inverse spectral gap of its continuous-time version (see Section \ref{app:continuous_time} in the supplement).

It can be shown that in cases where the importance tempering procedure coincides with classical MCMC (i.e.\ when $Z(\gamma)= 1$) the two definitions of relaxation times discussed above coincide.
}

\rev{
\subsubsection{Diagonal $X^TX$}\label{sec:diagonal}
Consider the case where all regressors are orthogonal to each other, i.e.\ $X^TX$ is diagonal. The latter requires $n\geq p$. 
The resulting posterior distribution for the inclusion variables $\gamma=(\gamma_1,\dots,\gamma_p)$ is a collection of independent Bernoulli random variables.
Denoting by $q_i$ the PIP of the $i$-th regressor, the posterior distribution of interest $f(\gamma)=p(\gamma|Y)$ has the following form
\begin{equation}\label{eq:target_iid}
f(\gamma)=\prod_{i=1}^p q_i^{\gamma_i}(1-q_i)^{1-\gamma_i}\,.
\end{equation}
Sampling from a target with independent components as in \eqref{eq:target_iid} is the ideal scenario for GS, and we are interested in understanding how suboptimal TGS and wTGS are compared to GS in this context.
}\rev{
The following theorem provides expressions for the relaxation times of GS, TGS and wTGS.
\begin{theorem}\label{thm:rel_times_diagonal}
Under \eqref{eq:target_iid}, the relaxation times of GS, TGS and wTGS satisfy
\begin{equation}\label{eq:rel_times_diagonal}
t_{GS}=\alpha_1p,\qquad
t_{TGS}= \alpha_2 p,\qquad
t_{wTGS}= s(1-q_{min})\,,
\end{equation}
where
$\alpha_1=\max\{q_{max},1-q_{min}\}$,
$\alpha_2=\max_{i\in\{1,\dots,p\}}q_i(1-q_i)$, $q_{max}=\max_{i\in\{1,\dots,p\}}q_i$, $q_{min}=\min_{i\in\{1,\dots,p\}}q_i$ and $s=\sum_{i=1}^pq_i$.
\end{theorem}
Theorem \ref{thm:rel_times_diagonal} implies that $t_{GS}$ and $t_{TGS}$ are proportional to the total number of variables $p$, while $t_{wTGS}$ depends only on the expected number of active variables $s=\sum_{i=1}^pq_i$, which is often much smaller than $p$.
Assuming $\alpha_2$ and $q_{min}$ to be bounded away from, respectively, $0$ and $1$ as $p\to\infty$; the results in \eqref{eq:rel_times_diagonal} imply that both GS and wTGS have $\mathcal{O}(pns)$ computational complexity, while the complexity of TGS is $\mathcal{O}(p^2n)$.
If $X^TX$ is precomputed before the MCMC run (see Section \ref{sec:efficient_impl}), the complexities are reduced to $\mathcal{O}(ps^2)$ for GS and wTGS and to $\mathcal{O}(p^2s)$ for TGS.
It follows that, even in the case of independent components, wTGS has the same theoretical cost of GS.
On the other hand, TGS is suboptimal by a $\mathcal{O}(p/s)$ factor.
\begin{remark}
The analysis above ignores Rao-Blackwellisation, which can be favourable to TGS and wTGS.
In fact, when $X^TX$ is diagonal the Rao-Blackwellised PIP estimators of TGS and wTGS are deterministic and return the $p$ exact PIPs in one iteration with cost $\mathcal{O}(np)$.
By contrast, GS has an $\mathcal{O}(nps)$ cost for each i.i.d.\ sample.
\end{remark}
}

\rev{
\subsubsection{Fully collinear case}\label{sec:collinear}
We now consider the other extreme case, where there are maximally correlated regressors.
In particular, suppose that $m$ out of the $p$ available regressors are perfectly collinear among themselves and with the data vector (i.e.\ each regressor fully explains the data), while the other $p-m$ regressors are orthogonal to the first $m$ ones.
For simplicity, assume $\Sigma_\gamma=c(X_\gamma^TX_\gamma)^{-1}$ and $h\in(0,1)$ fixed.
The $X^TX$ matrix resulting from the scenario described above is not full-rank.
In such contexts, the standard definition of $g$-priors, $\Sigma_\gamma=c(X_\gamma^TX_\gamma)^{-1}$, is not directly applicable and needs to be replaced by the more general definition involving generalised inverses (details in Section \ref{app:target_derivation} of the supplement).
}

\rev{
The posterior distribution of interest  $f(\gamma)=p(\gamma|Y)$ has the following structure
\begin{equation}\label{eq:target_collinear}
f(\gamma)
=
f_0(\gamma_1,\dots,\gamma_m)
\prod_{i=m+1}^p q_i^{\gamma_i}(1-q_i)^{1-\gamma_i}\,,
\end{equation}
where $q_i\in(0,1)$ is the posterior inclusion probability of the $i$-th variable for $i=m+1,\dots, p$ and $f_0$ denotes the joint distribution of the first $m$ variables.
By construction, the distribution $f_0$ is symmetric, meaning that 
$f_0(\gamma_1,\dots,\gamma_m)=q(\sum_{i=1}^m\gamma_i)$ for some $q:\{0,\dots,m\}\to[0,1]$\,.
See Section \ref{app:target_derivation} of the supplement for the specific form of $q$.
Under mild assumptions, we have $q(s)/q(1)\to 0$ as $p\to \infty$ for all $s\neq 1$, meaning that the distribution $f_0$ concentrates on the configurations having one and only one active regressor as $p$ increases.
}

\rev{
We study the asymptotic regime where $m$ is fixed and $p\to\infty$.
This corresponds to the commonly encountered scenario of having a small number of ``true'' variables and a large number of noise ones.
The latter has been the focus of much of the recent BVS literature \citep{johnson2012bayesian} and is motivated, for example, by applications to genomics (see examples in Section \ref{sec:real_data}).
In our analysis, the number of datapoints $n$, as well as the hyperparameters $c$ and $h$, can depend on $p$ in an arbitrarily manner, provided the following technical assumption is satisfied. 
\begin{ass}\label{ass:asympt_regime}
$\lim_{p\to\infty}h(1+c)^{-1/2} =0 $, $\limsup_{p\to\infty}h<1$ and $\liminf_{p\to\infty} h^2(1+c)^{(n-2)/2} >0$.
\end{ass}
Assumption \ref{ass:asympt_regime} is weak and satisfied in nearly any realistic scenario.
For example, it is satisfied whenever $c\geq 1$ and $h$ goes to 0 at a slower than exponential rate in $p$.
Note that the assumptions on $X^TX$ impose the constraint $n\geq p-m+1$.
}
 
\rev{
The following theorem characterises the behaviour of the relaxation times of GS, TGS and wTGS as $p$ increases.
\begin{theorem}\label{thm:rel_times_collinear}
As $p\to\infty$, the relaxation times of GS, TGS and wTGS satisfy
\begin{equation}\label{eq:rel_times_diagonal}
t_{GS}\geq \mathcal{O}(c^{1/2}h^{-1}p),\qquad
t_{TGS}\geq \mathcal{O}(p),\qquad
t_{wTGS}= \mathcal{O}(s)\,,
\end{equation}
where $s=\E_f[|\gamma|]$ is the expected number of active variables a posteriori.
\end{theorem}

Theorem \ref{thm:rel_times_collinear} implies that wTGS has $\mathcal{O}(pns)$ computational complexity, while TGS has complexity at least $\mathcal{O}(p^2n)$.
We conjecture $t_{TGS}\leq  \mathcal{O}(p)$ and we discuss a proof strategy in Remark \ref{rmk:conjecture} of the supplementary material.
If such conjecture is correct, then TGS has complexity exactly $\mathcal{O}(p^2n)$.
On the other hand, \eqref{eq:rel_times_diagonal} implies that the computational complexity of GS is at least $\mathcal{O}(pns c^{1/2}h^{-1})$, whose asymptotic behaviour depends on the choices of $c$ and $h$.
In general, wTGS provides an improvement over GS of at least $\mathcal{O}(c^{1/2}h^{-1})$.
If $h=\mathcal{O}(p^{-1})$ and $c=n$ such an improvement is at least $\mathcal{O}(pn^{1/2})$, while if $c=p^2$ it is at least $\mathcal{O}(p^2)$.

Theorems \ref{thm:rel_times_diagonal} and \ref{thm:rel_times_collinear} suggest that the relaxation times of TGS and wTGS are not significantly impacted by change in correlation structure between \eqref{eq:target_iid} and \eqref{eq:target_collinear}.
As discussed in Section \ref{sec:complexity_definition}, this supports the conjectures of Section \ref{sec:gaussians}.
}

\section{Simulation studies}\label{sec:simulations}
\rev{In this section we provide simulation studies illustrating the performances of GS, TGS and wTGS in the Bayesian Variable Selection (BVS) context described in Section \ref{sec:BVS}.}
\subsection{Illustrative example.}\label{sec:illustrative_BVS}
The differences between GS, TGS and wTGS can be well illustrated considering a scenario where two regressors with good explanatory power are strongly correlated.

In such a situation, models including one of the two variables will have high posterior probability, while models including both variables or none of the two will have a low posterior probability.
As a result, the Gibbs Sampler (GS) will get stuck in one of the two local modes corresponding to one variable being active and the other inactive.

Figure \ref{fig:BVS_illustrative_100} considers simulated data with $n=100$ and $p=100$, where the two correlated variables are number 1 and 2.
The detailed simulation set-up is described in Section \ref{sec:sim_data} (namely Scenario 1 with SNR=3).
\begin{figure}[h!]
\includegraphics[width=\linewidth]{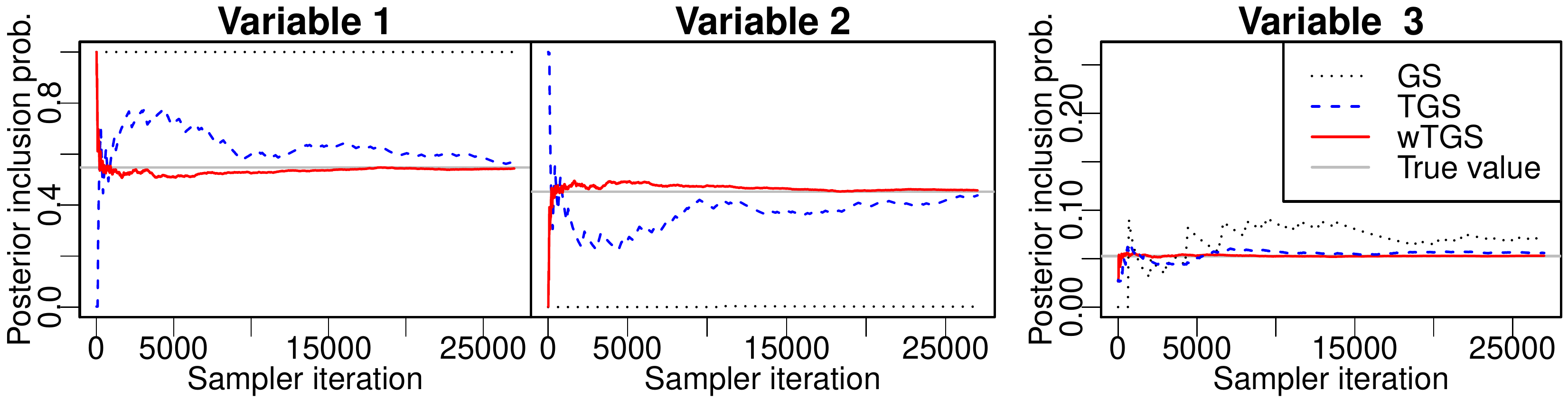}
\caption{\rev{Running estimates of PIPs for variables 1, 2 and 3 produced by GS, TGS and wTGS. Here $p=n=100$.
Thinning is used so that all schemes have the same cost per iteration.
The horizontal gray lines indicate accurate approximations to the true values of the PIPs.
}
}\label{fig:BVS_illustrative_100}
\end{figure}
\rev{All chains were started from the empty model ($\gamma_i=0$ for every $i$).
TGS and wTGS, which have a roughly equivalent cost per iteration, were run for $30000$ iterations, after a burn in of $5000$ iterations.
GS was run for the same CPU time, performing multiple moves per iteration so that the cost per iteration matched the one of TGS and wTGS.
The left and center plots in the figure display the traceplots of the estimates for the PIP of variables 1 and 2 for GS, TGS and wTGS.
The true PIP values are indicated with gray horizontal lines.
Such values are accurate approximation to the exact PIP obtained by running an extremely long run of wTGS.
For the purposes of this illustration, it is reasonable to treat these values as exact as the associated Monte Carlo error is orders of magnitude smaller then the other Monte Carlo errors involved in the simulation.
In the displayed run, GS got stuck in the mode corresponding to $(\gamma_1,\gamma_2)=(1,0)$ and never flipped variable 1 or 2. 
On the contrary, both TGS and wTGS manage to move frequently between the two modes and indeed the resulting estimates of PIPs for both variables appear to converge to the correct value, with wTGS converging significantly faster.}
It is also interesting to compare the schemes efficiency in estimating PIP for variables with lower but still non-negligible inclusion probability.
For example variable 3 in this simulated data has a PIP of roughly $0.05$. In this case the variable is rarely included in the model and the frequency-based estimators have a high variability, while the Rao-Blackwellised ones produce nearly-instantaneous good estimates, see Figure \ref{fig:BVS_illustrative_100} right.

Consider then an analogous simulated dataset with $p=1000$ and $n=500$.
\begin{figure}[h!]
\includegraphics[width=\linewidth]{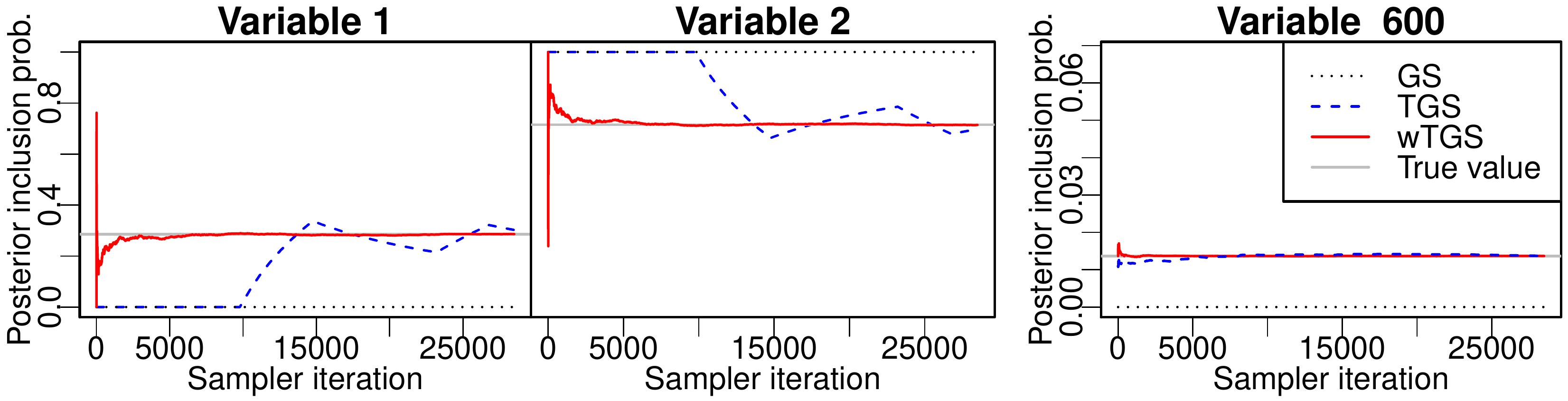}
\caption{Analogous to Figure 4 with $p=1000$ and $n=500$.
}\label{fig:BVS_illustrative_1000}
\end{figure}
In this case the larger number of regressors induces a more significant difference between TGS and wTGS as the latter focuses the computational effort on more important variables.
In fact, as shown in Figure \ref{fig:BVS_illustrative_1000}, both TGS and wTGS manage to move across the $(\gamma_1,\gamma_2)=(0,1)$ and $(\gamma_1,\gamma_2)=(1,0)$ modes but wTGS does it much more often and produce estimates converging dramatically faster to the correct values.
This is well explained by Proposition \ref{prop:wTGS}, which implies that TGS flips each variable every $1/p$ iterations on average, while wTGS has frequency of flipping equal to $\zeta^{-1}\E[\eta_i(\gamma_{-i})]$ defined in Section \ref{sec:wTGS_BVS}, which is a function of $p(\gamma_j =1|Y)$.
The faster mixing of wTGS for the most influential variables accelerates also the estimation of lower but non-negligible PIPs, such as coordinates 3 and 600 in Figures \ref{fig:BVS_illustrative_100} and \ref{fig:BVS_illustrative_1000}, respectively. 

To summarise, the main improvements of TGS and wTGS are due to:
\begin{itemize}[noitemsep,nolistsep]
\item[(i)] tempering reducing correlation and helping to move across modes (see Figure \ref{fig:BVS_illustrative_100} left and center);
\item[(ii)] Rao-Blackwellisation producing more stable estimators (see Figures \ref{fig:BVS_illustrative_100}-\ref{fig:BVS_illustrative_1000} right);
\item[(iii)] weighting mechanism of wTGS allowing to focus computation on relevant variables (see Figure \ref{fig:BVS_illustrative_1000} left and center).
\end{itemize}
The qualitative conclusions of this illustrative example would not change if one considers a scenario involving $m$ strongly correlated variables, with $m> 2$.

\subsection{Simulated data.}\label{sec:sim_data}
In this section we provide a quantitative comparison between GS, TGS and wTGS under different simulated scenarios.
Data are generated as $Y\sim N(X\beta^*,\sigma^2)$ with $\sigma^2=1$, $\beta^*=\hbox{SNR}\sqrt{\frac{\sigma^2 \log (p)}{n}} \beta^*_0$, and each row $(X_{i1},\dots,X_{ip})$ of the design matrix $X$ independently simulated from a multivariate normal distribution with zero mean and covariance $\Sigma^{(X)}$ having $\Sigma^{(X)}_{jj}=1$ for all $j$. 
We set the prior probability $h$ to $5/p$, corresponding to a prior expected number of active regressors equal to 5.
The values of $\beta^*_0$ and $\Sigma^{(X)}_{ij}$ for $i\neq j$ vary depending on the considered scenario.
In particular, we consider the following situations:
\begin{enumerate}[noitemsep,nolistsep]
\item[1.] \emph{Two strongly correlated variables:}
$\beta^*_0=(1,0,\dots,0)$, $\Sigma^{(X)}_{12}=\Sigma^{(X)}_{21}=0.99$, $\Sigma^{(X)}_{ij}=0$ otherwise.
\item[2.] \emph{Batches of correlated variables:}
$\beta^*_0=(3,3,-2,3,3,-2,0,\dots,0)$, 
$\Sigma^{(X)}_{ij}=0.9$ if $i,j\in\{1,2,3\}$ or  $i,j\in\{4,5,6\}$ and $\Sigma^{(X)}_{ij}=0$ otherwise.
\item[3.] \emph{Uncorrelated variables:}
$\beta^*_0=(2,-3,2,2,-3,3,-2,3,-2,3,0,\dots,0)$, $\Sigma^{(X)}_{ij}=0$ for all $i\neq j$.
\end{enumerate}
Scenarios analogous to the ones above have been previously considered in the literature.
For example, \citet[Sec.3.2.3]{TitsiasYau2017} consider a scenario similar to 1, \citet[Ex.4]{Wang2011} and \citet[Sec4.2]{Huang2016VB} one similar to 2 and \citet{Yang2016} one analogous to 3.
We compare GS, TGS and wTGS on all three scenarios for a variety of values of $n$, $p$ and SNR.
To have a fair comparison, we implement the Metropolised version of GS, like we did for TGS and wTGS.
In order to provide a quantitative comparison we consider a standard measure of relative efficiency, being the ratio of the estimators' effective sample sizes over computational times.
More precisely, we define the relative efficiency of TGS over GS as
\begin{equation}\label{eq:rel_eff}
\frac{\hbox{Eff}_{TGS}}{\hbox{Eff}_{GS}}
=
\frac{\hbox{ess}_{TGS}/T_{TGS}}{\hbox{ess}_{GS}/T_{GS}}
=
\frac{\sigma^2_{GS}T_{GS}}{\sigma^2_{TGS}T_{TGS}}
\,,
\end{equation}
where $\sigma^2_{GS}$ and $\sigma^2_{TGS}$ are the variances of the Monte Carlo estimators produced by $GS$ and $TGS$, respectively, while $T_{GS}$ and $T_{TGS}$ are the CPU time required to produce such estimators.
An analogous measure is used for the relative efficiency of wTGS over GS.
For each simulated dataset, we computed the relative efficiency defined by \eqref{eq:rel_eff} for each PIP estimator, thus getting $p$ values, one for each variable. Table \ref{table:median} reports the median of such $p$ values for each dataset under consideration.
\rev{The variances in \eqref{eq:rel_eff}, such as $\sigma^2_{GS}$ and $\sigma^2_{TGS}$, were estimated with the sample variances of the PIP estimates obtained with 50 runs of each algorithm. See Section \ref{app:simulations} of the supplement for more details.
}
\begin{table}
\caption{\label{table:median}Median improvement over variables of TGS and wTGS relative to GS for simulated data. Scenarios 1 to 3, indicated on the leftmost column, are described in Section \ref{sec:sim_data}. Notation: 1.4e5$=1.4\times 10^5$.
}
\centering
\begin{tabular}{c|c|cccc|cccc}
  \hline
 & &  &  \multicolumn{2}{c}{TGS-vs-GS}   & &  &   \multicolumn{2}{c}{wTGS-vs-GS}  &  \\ 
   \hline
 & &  & \multicolumn{2}{c}{SNR}  & &  & \multicolumn{2}{c}{SNR} &  \\ 
  \hline
 &(p,n) & 0.5 & 1 & 2 & 3 & 0.5 & 1 & 2 & 3 \\ 
  \hline
\parbox[t]{2mm}{\multirow{3}{*}{\rotatebox[origin=c]{90}{scen.1}}}  & (100,50) & 
4.0e5 & 2.4e4 & 2.0e4 & 6.6e4 & 2.1e6 & 2.6e5 & 3.4e5 & 1.9e5 \\ 
 & (200,200) & 
1.0e6 & 4.2e6 & 4.9e5 & 2.1e6 & 1.6e7 & 5.3e7 & 1.0e7 & 2.4e7 \\ 
 & (1000,500) & 
1.3e6 & 1.2e6 & 1.1e6 & 2.2e6 & 7.8e7 & 9.3e7 & 6.5e7 & 1.1e8 \\ 
  \hline
\parbox[t]{2mm}{\multirow{3}{*}{\rotatebox[origin=c]{90}{scen.2}}}    &  (100,50) & 
1.0e4 & 2.9e3 & 1.7e3 & 3.9e4 & 1.5e5 & 4.1e4 & 9.3e3 & 1.6e5 \\ 
 & (200,200) & 
1.1e5 & 1.0e5 & 8.2e3 & 1.4e7 & 1.8e6 & 2.8e6 & 1.5e5 & 3.2e6 \\ 
 & (1000,500) & 
4.6e5 & 9.2e4 & 6.7e5 & 2.1e6 & 3.3e7 & 1.1e7 & 1.1e7 & 1.5e7 \\ 
\hline
\parbox[t]{2mm}{\multirow{3}{*}{\rotatebox[origin=c]{90}{scen.3}}}  & (100,50) & 
2.5e3 & 4.2e3 & 7.7e3 & 7.4e4 & 2.9e4 & 3.9e4 & 8.0e3 & 1.5e4 \\ 
 & (200,200) & 
9.1e4 & 4.3e4 & 2.8e7 & 3.5e6 & 1.0e6 & 3.1e5 & 2.9e6 & 8.0e5 \\ 
 & (1000,500) & 
9.8e4 & 5.9e5 & 1.1e7 & 2.1e7 & 7.0e6 & 4.4e6 & 7.6e6 & 1.0e7 \\ 
\end{tabular}
\end{table}

From Table \ref{table:median} it can be seen that both TGS and wTGS provide orders of magnitude improvement in efficiency compared to GS, with median improvement of TGS over GS ranging from $1.7\times 10^3$ to $2.1\times 10^7$ and of wTGS over GS ranging from $8.0\times 10^3$ to $1.1\times 10^8$.
Such a huge improvement, however, needs to be interpreted carefully.
In fact, in all simulated datasets the fraction of variables having non-negligible PIP is small (as it is typical in large $p$ BVS applications) and thus the median improvement refers to the efficiency in estimating a variable with very small PIP, e.g.\ below $0.001$.
When estimating such small probabilities, standard Monte Carlo estimators perform poorly compared to Rao-Blackwellised versions (see right of Figures \ref{fig:BVS_illustrative_100} and \ref{fig:BVS_illustrative_1000}) and this explains such a huge improvement of TGS and wTGS over GS. %
In many practical scenarios, however we are not interested in estimating the actual value of such small PIP.
Thus a more informative comparison can be obtained by restricting our attention to variables with moderately large PIP.
Table \ref{table:mean_large_PIP} reports the mean relative efficiency for variables whose PIP is estimated to be larger than 0.05 by at least one of the algorithms under consideration.
\begin{table}
\caption{\label{table:mean_large_PIP}Mean improvement of TGS and wTGS relative to GS over variables with PIP$>$0.05. 
Same simulation set-ups as in Table \ref{table:median}.
Empty values corresponds to large values with no reliable estimate available (see Section \ref{sec:sim_data} for discussion).
}
\centering
\begin{tabular}{c|c|cccc|cccc}
  \hline
 &  &  &  \multicolumn{2}{c}{TGS-vs-GS}   & &  &   \multicolumn{2}{c}{wTGS-vs-GS}  &  \\ 
   \hline
 &  &  & \multicolumn{2}{c}{SNR}  & &  & \multicolumn{2}{c}{SNR} &  \\ 
  \hline
 & (p,n) & 0.5 & 1 & 2 & 3 & 0.5 & 1 & 2 & 3 \\ 
  \hline
\parbox[t]{2mm}{\multirow{3}{*}{\rotatebox[origin=c]{90}{scen.1}}} & (100,50) & 
  & 7.2e1 & 1.8e1 & 2.8e2 &  & 5.8e2 & 4.2e2 & 3.1e3 \\ 
 & (200,200) & 
4.9e3 &  & 6.6e1 & 1.9e2 & 1.1e4 &  & 1.8e3 & 1.6e4 \\ 
 & (1000,500) & 
2.7e2 & 6.3e2 & 1.4 & 8.1e1 & 8.8e3 & 2.5e4 & 5.8e2 & 1.9e4 \\ 
   \hline
 \parbox[t]{2mm}{\multirow{3}{*}{\rotatebox[origin=c]{90}{scen.2}}} &   (100,50) & 
 4.8 & 1.4e1 & 3.3 & 2.0e1 & 1.3e2 & 2.4e2 & 1.8e1 & 1.4e2 \\ 
 & (200,200) & 
  8.6e1 & 4.7e1 & 3.4 & 2.5e6 & 2.3e3 & 2.1e3 & 6.0e1 & 4.1e2 \\ 
 & (1000,500) & 
  4.6e1 & 3.7e1 & 1.3e1 & 4.5e2 & 1.1e4 & 7.6e3 & 1.1e3 & 1.8e4 \\ 
  \hline
\parbox[t]{2mm}{\multirow{3}{*}{\rotatebox[origin=c]{90}{scen.3}}} &   (100,50) & 
2.7 & 5.3 & 9.2 &  & 2.5e1 & 6.7e1 & 2.1e1 &  \\ 
 & (200,200) & 
1.1e2 & 6.6e1 &  &  & 1.3e3 & 4.6e2 &  &  \\ 
 & (1000,500) & 
1.6e1 & 6.8e2 &  &  & 1.1e3 & 9.4e3 &  & 
\end{tabular}
\end{table}
Empty values correspond to cells where either no PIP was estimated above $0.05$ or where GS never flipped such variable and thus we had no natural (finite) estimate of the variance in \eqref{eq:rel_eff}.
In both such cases we expect the improvement in relative efficiency over GS to be extremely large (either corresponding to the values in Table 1, first case, or currently estimated at infinity, second case) and thus excluding those values from Table \ref{table:mean_large_PIP} is conservative and plays in favour of GS.
The mean improvements reported in Table \ref{table:mean_large_PIP} are significantly smaller than the one in Table \ref{table:median} but still potentially very large, with ranges of improvement being $(1.4,2.5\times 10^6)$ for TGS and$(1.8\times 10^1,1.9\times 10^4)$ for wTGS.
Note that there is no value below 1, meaning that in these simulations TGS or wTGS are always more efficient than GS, and that wTGS is more efficient than TGS in most scenarios.
Also, especially for wTGS, the improvement over GS gets larger as $p$ increases.

The value of $c$ in the prior covariance matrix has a large impact on the concentration of the posterior distribution and thus on the resulting difficulty of the computational task.
Different suggestions for the choice of $c$ have been proposed in the literature, such as $c=n$ \citep{Zellner1986}, $c=\max\{n,p^2\}$
\citep{fernandez2001benchmark} or a fixed value between $10$ and $10^4$ \citep{smith1996nonparametric}.
For the simulations reported in Tables \ref{table:median} and \ref{table:mean_large_PIP} we set $c=10^3$, which provided results that are fairly representative in terms of relative efficiency of the algorithms considered.
In Section \ref{sec:real_data} we will consider both $c=n$ and $c=\max\{n,p^2\}$.

\subsection{Real data.}\label{sec:real_data}
In this section we consider three real datasets with increasing number of covariates.
We compare wTGS  to GS and the Hamming Ball (HB) sampler, a recently proposed sampling scheme designed for posterior distributions over discrete spaces, including BVS models \citep{TitsiasYau2017}.
We consider three real datasets, which we refer to as DLD data, TGFB172 data and TGFB data.
The DLD data comes from a genomic study by \citet{yuan2016plasma} based on RNA sequencing and has a moderate number of regressors, $p=57$ and $n=192$.
The version of the dataset we used is freely available from the supplementary material of \citet{rossell2017tractable}. 
See Section 6.5 therein for a short description of the dataset and the inferential questions of interest.
The second and third datasets are human microarray gene expression data in colon cancer patients from \citet{calon2012dependency}.
The TGFB172 data, which has $p=172$ and $n=262$, is obtained as a subset of the TGFB data, for which $p=10172$ and $n=262$.
These two datasets are are described in Section 5.3 of \citet{rossell2017nonlocal} and are freely available from the corresponding supplementary material.

If $X^TX$ and $Y^TX$ are precomputed, the cost per iteration of the algorithms under consideration is not sensitive to $n$ (see Section \ref{sec:efficient_impl} and Section \ref{appendix:cost_iter} of the supplement). 
Thus a dataset with a large value of $p$, like the TGFB data, represents a computationally challenging scenario, regardless of having a low value of $n$.
\rev{Moreover, low values of $n$ have been reported to induce posterior distributions $p(\gamma|Y)$ that are less concentrated and harder to explore \citep[Sec.3-4]{johnson2013numerical}. 
In this sense, small-$n$-large-$p$ scenarios are among the most computationally challenging ones in the BVS scenario.}


We performed 20 independent runs of each algorithm for each dataset with both $c=n$ and $c=p^2$, recording the resulting estimates of PIPs.
We ran wTGS for 500, 1000 and 30000 iterations for the DLD, TGFB172 and TGFB datasets, respectively, discarding the first 10$\%$ of samples as burnin.
The number of iterations of GS and HBS were chosen to have the same runtime of wTGS.
To assess the reliability of each algorithm, we compare results obtained over different runs by plotting each PIP estimate over the ones obtained with different runs of the same algorithm.
The results are displayed in Figure \ref{fig:diagonal_all_datasets}.
\begin{figure}[h!]
\includegraphics[width=\linewidth]{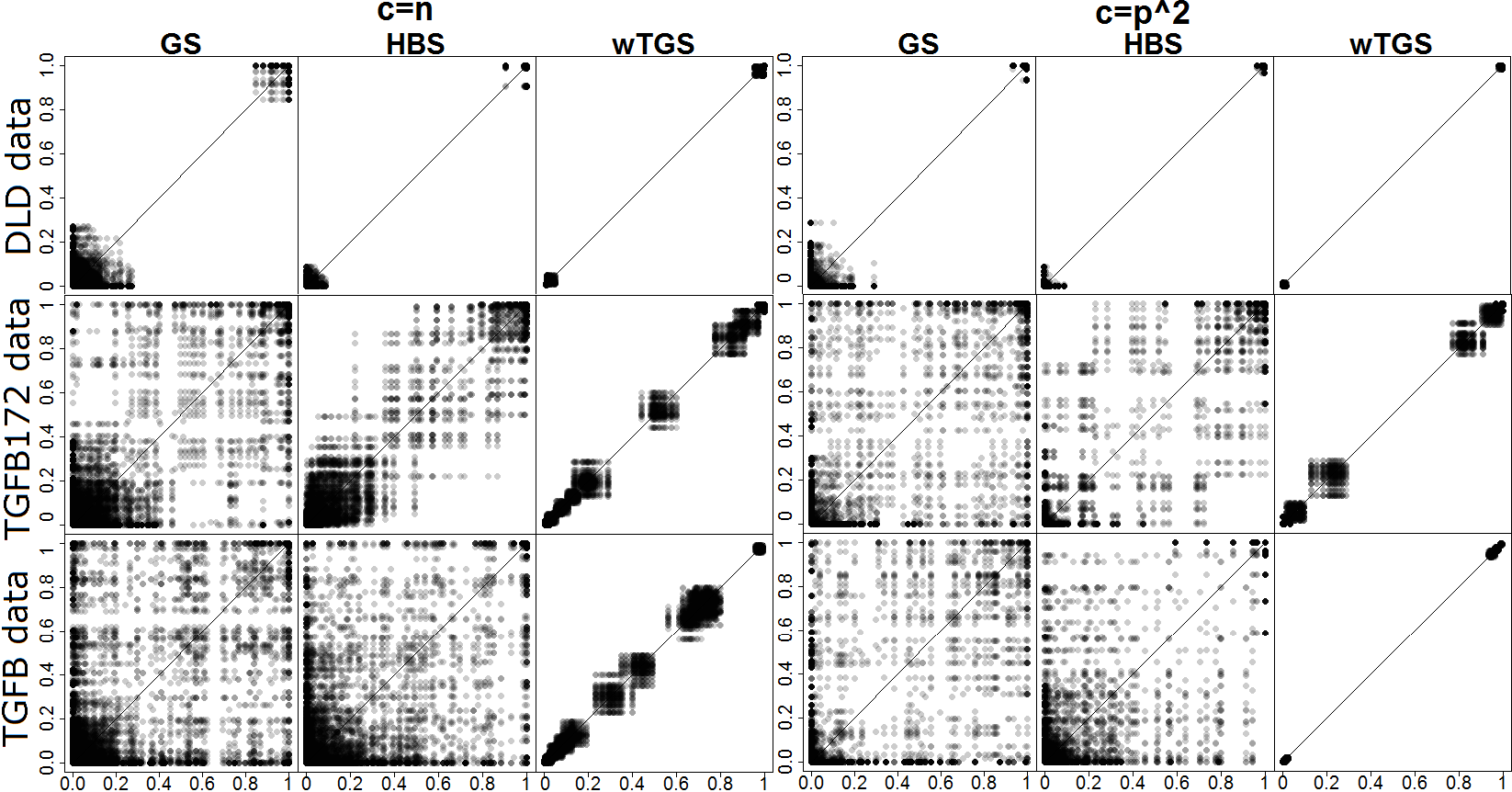}
\caption{Comparison of GS, HBS and wTGS (columns) on three real datasets (rows) for $c=n$ and $c=p^2$.
Points close to the diagonal lines indicate estimates agreeing across different runs.}\label{fig:diagonal_all_datasets}
\end{figure}
Points close to the diagonal indicate estimates in accordance with each other across runs, while point far from the diagonal indicate otherwise.
It can be seen that wTGS provides substantially more reliable estimates for all combinations of dataset and value of $c$ under consideration and that the efficiency improvement increases with the number of regressors $p$.
\rev{Since each box in Figure \ref{fig:diagonal_all_datasets} contains a large number of PIP estimates (namely $p\times 20\times 19$ points), we also provide the analogous figure obtained by running only two runs of each algorithm in Section \ref{sec:add_figure} of the supplement.
The latter representation may be more familiar to the reader.
}

All computations reported in Section \ref{sec:simulations} were performed on the same desktop computer with 16GB of RAM and an i7 Intel processor, using the $R$ programming language \citep{R}.
The $R$ code to implement the various samplers under consideration is freely available at \url{https://github.com/gzanella/TGS}.
For the largest dataset under consideration (p=10172) wTGS took an average of 115 seconds for each run shown in Figure \ref{fig:diagonal_all_datasets}.
We performed further experiments, in order to compare the wTGS performances with the ones of available $R$ packages for BVS and some alternative methodology from the literature.
The results, reported in Section \ref{appendix:BVS_lit_review} of the supplement, suggest that wTGS provides state of the art performances for fitting spike and slab BVS models like the ones of Section \ref{sec:BVS_model}.

\section{Discussion}\label{sec:discussion}

We have introduced a novel Gibbs sampler variant, demonstrating its considerable potential both in toy examples as well as more realistic Bayesian Variable Selection models, and giving underpinning theory to support the use of the method and to explain its impressive convergence properties.

TGS can be thought of as an intelligent random scan Gibbs sampler, using current state information to inform the choice of component to be updated. In this way, the method is different from the usual random scan method which can also have heterogeneous component updating probabilities which can be optimised (for example by adaptive MCMC methodology, see for example \citealp{chimisov2018adapting}).


There are many potential extensions of TGS that we have not considered in this paper. For example, we could replace Step 2 of TGS, where $i$ is sampled proportionally to $p_i(\x)$, with a Metropolised version as in \citep{liu1996peskun}, where the new value $i^{(t+1)}$ is proposed from $\{1,\dots,d\}\backslash \{i^{(t)}\}$ proportionally to $p_{i^{(t+1)}}(\x)$ for $i^{(t+1)}\neq i^{(t)}$. This would effectively reduce the probability of repeatedly
 updating the same coordinate in consecutive iterations, which, as shown in Proposition \ref{prop:bivariate_exchangeable}, can be interpreted as a rejected move.

Another direction for further research might aim to reduce the cost per iteration of TGS when $d$ is very large.
For example, we could consider a ``block-wise'' version of TGS, where first a subset of variables is selected at random and then TGS is applied only to such variables conditionally on the others, to avoid computing all the values of $\{p_i(\x)\}_{i=1}^d$ at each iteration. The choice of the number of variables to select would then be related to a cost-per-iteration versus mixing trade-off. 
See Section 6.4 of \citet{zanella2017informed} for a discussion of similar block-wise implementations.
Also, computing $p_i(\x)$ exactly may be infeasible in some contexts, and thus it would be interesting to design a version of TGS where the terms $p_i(\x)$ are replaced by unbiased estimators while preserving the correct invariant distribution.

A further possibility for future research is to construct deterministic scan versions of TGS which may be of value for contexts where deterministic scan Gibbs samplers are known to outperform random scan ones (see for example \citealp{roberts2015surprising}).
\rev{Also, it would be useful to provide detailed methodological guidance regarding the choice of good modified conditionals $g(x_i|\x_{-i})$, e.g.\ good choices of the tempering level $\beta$, extending the preliminary results of Section \ref{sec:gaussians}.}

One could design schemes where the conditional distributions of $k$ coordinates are tempered at the same time, 
rather than a single coordinate. 
A natural approach would be to use the TGS interpretation of Remark \ref{rmk:extended_target} and define some extended target on $\sX\times\{1,\dots,d\}^k$.
This would allow to achieve good mixing in a larger class of target distributions (compared to the ones of Section \ref{sec:gaussians}) at the price of a larger cost per iteration.

TGS provides a generic way of mitigating the worst effects of dependence on Gibbs sampler convergence.
Classical ways of reducing posterior correlations involve reparametrisations \citep{gelfand1995efficient,hills1992parameterization}. Although these can work very well in some specific models (see e.g. \citealp{zanella2017analysis,papaspiliopoulos2018scalable}), the generic implementations requires the ability to perform Gibbs Sampling on generic linear transformations of the target, which is often not practical beyond the Gaussian case. For example it is not clear how to apply such methods to the BVS models of Section \ref{sec:BVS}. Moreover reparametrisation methods are not effective if the covariance structure of the target changes with location.
Further alternative methodology to overcome strong correlations in Gibbs Sampling include the recently proposed adaptive MCMC approach of \cite{duan2017calibrated} in the context of data augmentation models.

Given the results of Sections \ref{sec:BVS} and \ref{sec:simulations}, it would be interesting to explore the use of the methodology proposed in this paper for other BVS models, such as models with more elaborate priors (e.g. \citealp{johnson2012bayesian}) or binary response variables.

\subsection*{Acknowledgments}
GZ supported by the European Research Council (ERC) through StG ``N-BNP'' 306406. 
GOR acknowledges support from EPSRC through grants EP/K014463/1 (i-Like) and EP/K034154/1 (EQUIP).

\bibliographystyle{rss}
\bibliography{bibliography_TGS}
\end{document}